\documentclass[a4paper, UKenglish, cleveref, autoref, thm-restate]{lipics-v2021}

\hideLIPIcs{}
\usepackage{hyperref}
\usepackage{tabularx, environ}
\usepackage[linguistics]{forest}
\usepackage{verbatim}
\usepackage{amsfonts, amsmath, amscd, stmaryrd}
\usepackage{dsfont}
\usepackage{xcolor}
\usepackage{todonotes}
\usepackage{subcaption}
\usepackage{wrapfig}
\usepackage{float}
\usepackage{tikz}
\usetikzlibrary{shapes.geometric}
\usetikzlibrary{fit}
\usetikzlibrary{backgrounds}
 
\makeatletter
\newcommand{\problemtitle}[1]{\gdef\@problemtitle{#1}}
\newcommand{\probleminput}[1]{\gdef\@probleminput{#1}}
\newcommand{\problemquestion}[1]{\gdef\@problemquestion{#1}}
\NewEnviron{problemDefinition}{
  \problemtitle{}\probleminput{}\problemquestion{}
  \BODY{}
  \par\addvspace{.5\baselineskip}
  \noindent
  \begin{tabularx}{\textwidth}{|l|X|c}
    \hline
    \multicolumn{2}{|l|}{\@problemtitle} \\
    \hline
    \textbf{Input:} & \@probleminput{} \\
    \textbf{Question:} & \@problemquestion{} \\
    \hline
  \end{tabularx}
  \par\addvspace{.5\baselineskip}
}

\@namedef{T1/zi4/m/it}{<->ssub*zi4/m/n}
\@namedef{T1/zi4/b/it}{<->ssub*zi4/b/n}
\@namedef{T1/zi4/bx/it}{<->ssub*zi4/b/n}
\@namedef{T1/zi4/m/sl}{<->ssub*zi4/m/n}
\@namedef{T1/zi4/b/sl}{<->ssub*zi4/b/n}
\@namedef{T1/zi4/bx/sl}{<->ssub*zi4/b/n}
\@namedef{OT1/cmr/m/scit}{<->ssub*cmr/m/sc}
\@namedef{OT1/cmtt/m/scit}{<->ssub*cmtt/m/n}
\@namedef{OMS/cmtt/m/n}{<->ssub*cmsy/m/n}
\@namedef{U/stmry/b/n}{<->ssub*stmry/m/n}
\@namedef{T1/lmtt/m/scit}{<->ssub*lmtt/m/n}
\@namedef{T1/lmss/bx/sc}{<->ssub*lmss/bx/n}
\DeclareFontShape{T1}{lmr}{m}{scit}{<->ssub * lmr/m/scsl}{}
\makeatother

\title{Type-Checking for Pattern-Based Tree Transformations} 

\author{C. Aiswarya}{Chennai Mathmatical Institute, India \and CNRS IRL ReLaX, India }{aiswarya@cmi.ac.in}{0000-0002-4878-7581}{}

\author{Sahil Mhaskar\footnote{Corresponding author.}}{Chennai Mathematical Institute, India}{sahil@cmi.ac.in}{0000-0001-9650-290X}{}

\author{M. Praveen}{Chennai Mathmatical Institute, India \and CNRS IRL ReLaX, India}{praveenm@cmi.ac.in}{0000-0002-5734-7115}{}

\authorrunning{C. Aiswarya, S. Mhaskar, and M. Praveen}

\Copyright{C. Aiswarya, Sahil Mhaskar, M. Praveen}



\relatedversion{Full version of a preprint accepted at FSTTCS 2026}

\nolinenumbers

\newcommand{\twopartfunc}[5]{
    #1 := 
    \left\{ \begin{array}{ll}
        #2 & \mbox{ if } #3 ; \\
        #4 & \mbox{ if } #5
    \end{array} \right.
}

\newcommand{\N}{\mathbb{N}}
\newcommand{\R}{\mathcal{R}}
\newcommand{\trees}[1]{\ensuremath{\texttt{Trees}\left(#1\right)}}
\newcommand{\contexts}[2]{\ensuremath{\texttt{Contexts}_{#2}\left(#1\right)}}
\newcommand{\relations}{\texttt{Relations}}
\newcommand{\realisableRelations}[2]{\ensuremath{\texttt{IndRel}\left(#1, #2\right)}}
\newcommand{\Positions}[1]{\ensuremath{\texttt{Positions}(#1)}}

\newcommand{\src}{\mathit{src}}
\newcommand{\tgt}{\mathit{tgt}}
\newcommand{\genericMap}{\phi}

\newcommand{\bigo}{\mathcal{O}}
\newcommand{\np}{\textsc{NP}}

\newcommand{\nptc}{\textsc{NP}-complete}
\newcommand{\expt}{\textsc{EXPTime}}
\newcommand{\expth}{\textsc{EXPTime}-hard}
\newcommand{\exptc}{\textsc{EXPTime}-complete}

\newcommand{\matchingProblemFull}{\texttt{Matching Problem}}

\newcommand{\metaTypecheckingProblemFull}{\texttt{Type-Checking Problem}}
\newcommand{\intersectionNonemptinessFull}{\texttt{Intersection Nonemptiness of DBTA}}
\newcommand{\oneInThreeSatFull}{\texttt{1-in-3 SAT}}
\newcommand{\relationRealisabilityFull}{\texttt{Relation Realisability Problem}}

\newcommand{\metaTypecheckingProblemReformulationFull}{\texttt{Transformation Reachability Problem}}
\newcommand{\atomicTransformReachabilityFull}{\texttt{Atomic Transformation Reachability Problem}}

\newcommand{\blue}[1]{\textcolor{blue}{#1}}

\newcommand{\grey}[1]{\textcolor{gray}{#1}}
\newcommand{\tuple}[1]{\ensuremath{\overline{\textbf{#1}}}}

\newcommand{\regularOver}[1]{\ensuremath{\texttt{Regular}(#1)}}
\newcommand{\languageOf}[1]{\ensuremath{L(#1)}}

\newcommand{\seedLanguage}{S}
\newcommand{\genericLanguage}{L}
\newcommand{\sourceLanguage}{\ensuremath{\genericLanguage_{\src}}}
\newcommand{\targetLanguage}{\ensuremath{\genericLanguage_{\tgt}}}

\newcommand{\exampleSeedLanguageLoop}{\mathsf{loop}}
\newcommand{\exampleSeedLanguageEndLoop}{\mathsf{endLoop}}
\newcommand{\exampleSeedLanguageNull}{\mathsf{null}}

\newcommand{\examplePatternLanguageFor}{\mathsf{for}}
\newcommand{\examplePatternLanguageWhile}{\mathsf{while}}
\newcommand{\examplePatternLanguageInit}{\mathsf{init}}
\newcommand{\examplePatternLanguageCond}{\mathsf{cond}}
\newcommand{\examplePatternLanguageIncr}{\mathsf{incr}}
\newcommand{\examplePatternLanguageNull}{\mathsf{null}}

\newcommand{\transitionSystem}{\ensuremath{\mathcal{T}}}
\newcommand{\automaton}{\ensuremath{\mathcal{A}}}
\newcommand{\automatonB}{\ensuremath{\mathcal{B}}}
\newcommand{\uniqueLetterAutomaton}[1]{\ensuremath{\mathcal{U}_{#1}}}
\newcommand{\sourceAutomatonB}{\ensuremath{\automatonB_{\src}}}
\newcommand{\targetAutomatonB}{\ensuremath{\automatonB_{\tgt}}}
\newcommand{\extendedTransitionRelation}{ETR}
\newcommand{\run}{\rho}

\newcommand{\automatonStateSet}{Q}
\newcommand{\automatonTransitionsSet}{\Delta}
\newcommand{\automatonInitialStatesSet}{I}
\newcommand{\automatonFinalStatesSet}{F}
\newcommand{\stateSize}{\mathsf{stateSize}}

\newcommand{\alphabet}{\ensuremath{\Sigma}}
\newcommand{\coreAlphabet}{\ensuremath{\alphabet_\text{core}}}
\newcommand{\seedAlphabet}{\ensuremath{\Gamma}}
\newcommand{\variables}{\texttt{Var}}

\newcommand{\sourceMap}{\ensuremath{\genericMap_{\src}}}
\newcommand{\targetMap}{\ensuremath{\genericMap_{\tgt}}}
\newcommand{\substitution}{\ensuremath{\sigma}}

\newcommand{\guardFunction}{\ensuremath{\mathcal{G}}}
\newcommand{\guardOf}[1]{\ensuremath{\guardFunction(#1)}}

\newcommand{\atomicTransform}{\ensuremath{T}}
\newcommand{\fullTransform}{\ensuremath{\mathbf{T}}}

\newcommand{\pattern}{\mathsf{Pat}}
\newcommand{\sourcePattern}{\pattern_\src}
\newcommand{\targetPattern}{\pattern_\tgt}

\newcommand{\destinationFunction}[2]{\ensuremath{D_{#1}(#2)}}

\newcommand{\sem}[1]{\ensuremath{\llbracket #1 \rrbracket}}
\newcommand{\pt}[1]{\ensuremath{#1}-Pebble Tree Transducer}
\newcommand{\set}[1]{\ensuremath{\left\{ #1 \right\}}}
\newcommand{\abs}[1]{\ensuremath{\left| #1 \right|}}
\newcommand{\rev}{\texttt{Reverse}}
\DeclareMathOperator{\rank}{rank}  
\newcommand{\inv}{\ensuremath{^{-1}}}
\newcommand{\dom}{\textsf{Domain}}
\newcommand{\arity}{\texttt{arity}}
\newcommand{\operator}{\texttt{op}}
\newcommand{\result}{\texttt{result}}
\newcommand{\Results}{\texttt{Results}}
\newcommand{\shape}{\texttt{Shape}}
\newcommand{\relT}{\mathit{Rel}}
\newcommand{\start}{\mathsf{start}}
\newcommand{\e}{\varepsilon}
\newcommand{\tcp}{type-checking problem}
\newcommand{\fsa}{far-sighted automaton}

\newcommand{\B}{\mathcal{B}}
\newcommand{\Z}{\mathcal{Z}}
\newcommand{\D}{\mathcal{D}}

\newcommand{\Transitions}{\ensuremath{\Delta}}

\begin{document}

\maketitle


\begin{abstract}
    We introduce and study pattern-based tree transformations. As an illustrating example, consider a source pattern $(x \cdot y) + (x \cdot z)$ and a target pattern $x \cdot (y + z)$ as a pair. This source pattern matches any expression $e$ of the form $(e_1 \cdot e_2) + (e_1 \cdot e_3)$ (by substituting $x$ with $e_1$, $y$ with $e_2$, and $z$ with $e_3$) and the pair transforms it into the expression $e_1 \cdot (e_2 + e_3)$ as dictated by the target pattern. Note that in this example, the set of expressions that match the source pattern is not a regular tree language.

    We propose a model of tree transformations given by a finite representation of a (possibly infinite) set of such (source pattern, target pattern) pairs. The expressive power of this model comes at the cost of undecidability of checking equivalence. Nevertheless, we show that the type-checking problem is decidable for our model of pattern-based tree transformations. The type-checking problem asks whether applying a given transformation to trees having a given regular property (type) preserves the property. Our decision procedure is by a reduction to the emptiness problem of alternating tree automata.
\end{abstract}


\section{Introduction}\label{sec:introduction}

Tree transformations in various forms of machines and logics have been studied for applications such as XML processing, programming languages, and HTML sanitisation~\cite{DVLM2015,Hosoya_2010,AK1998,MSV2000}. Classical formalisms for specifying tree transductions such as finite-state top-down and bottom-up tree transducers, macro tree transducers, attribute grammars and MSO (Monadic Second-Order Logic)-definable graph transductions typically aim for robustness with respect to operations involving regular tree languages and decidability of the equivalence problem.

We introduce a model for tree transductions targeting scenarios where the domain of the transduction is not necessarily regular, and undecidability of equivalence does not render the model useless. Expanding on the example from the abstract, consider expressions of the form $t \cdot e_1 + t \cdot e_2 + \cdots + t \cdot e_n$, where $n$ is any number. An equivalent expression is $t \cdot (e_1 + e_2 + \cdots + e_n)$. Transformations like this are routinely used in Horner's method, an efficient way to evaluate polynomials. We consider the corresponding expression trees and model this transformation by matching the source tree with a source pattern and the target tree with a target pattern. The patterns are themselves trees, which can have placeholder variables to represent subtrees, like $t$ in the above example, which must repeat in multiple positions. The number $n$ can be arbitrary, so we need infinitely many pairs of patterns. For the example above, for each $n > 0$ we need $(x \cdot e_1 + x \cdot e_2 + \cdots + x \cdot e_n\,, \quad x \cdot (e_1 + e_2 + \cdots + e_n)\,)$ as a (source pattern, target pattern) pair, where $x$ is a placeholder variable intended to match any expression $t$. We represent such infinite sets (of pairs of patterns) using (standard) tree transductions. We allow placeholder variables in patterns to be substituted with any tree, as long as it satisfies a guard condition associated with the variable. The guard condition is given by a regular tree language called the \emph{guard language}. It allows, for example, to restrict the expression $t$ above to be a valid expression. In some situations we may want to restrict it further, say, to not contain any products inside. In this example, the set of source trees to which the transformation is applicable needs to have the same subtree $t$ repeat in multiple positions. Thus the set of source trees on which transformations are defined need not be a regular tree language in general, which makes our model different from classical tree transductions.

The backbone of our model is a finite representation of the infinite set of (source pattern, target pattern) pairs. Such a finite representation can be any formalism for tree transductions; in this paper we have chosen to use bimorphisms for technical convenience. Bimorphisms are given by a regular seed language $\seedLanguage$ and two morphisms $\sourceMap, \targetMap$. The transduction is defined as the relation $\{(\sourceMap(t), \targetMap(t)) \mid t \in \seedLanguage\}$. Over words, bimorphisms are equivalent to rational transductions~\cite{Nivat68}. Over trees, bimorphisms capture subclasses of tree transducers (cf.~\cite[Section 6.5]{tata} and~\cite{Engelfriet75}).

For example, consider the seed language given by $a b 1{(0+1)}^\ast {(\oplus b 1{(0+1)}^\ast)}^\ast c$. This is a regular word language over the alphabet $\set{a, b, c, \oplus, 0, 1}$. We define the \emph{source} and \emph{target} morphisms $\sourceMap, \targetMap$: $\set{a, b, c, \oplus, 0, 1} \to \set{0,1, +, \cdot, x, (, )}^\ast$  as given by the following: $\sourceMap := \set{a \mapsto \e, b \mapsto x \cdot {}, c \mapsto \e, \oplus \mapsto +, 0 \mapsto 0, 1 \mapsto 1}$ and similarly we define: $\targetMap := \set{a \mapsto x \cdot (, b \mapsto \e, c \mapsto ), \oplus \mapsto +, 0 \mapsto 0, 1 \mapsto 1}$. We note that this bimorphism generates pairs of the form $(x \cdot e_1 + x \cdot e_2 + \cdots + x \cdot e_n\,, \quad x \cdot (e_1 + e_2 + \cdots + e_n))$ where $e_i$ is a nonzero binary number, from the seed word $a b e_1 \oplus b e_2 \oplus \dots \oplus b e_n c$. We may add a guard to $x$ saying that $x$ can take values only from $1{(0 + 1)}^\ast$, capturing completely the distributivity of multiplication over addition for binary numbers. Note that this only requires bimorphisms over words and defines transduction over words. If instead we allow the guard language of $x$ to be any syntactically valid expression tree, this now defines transformations from trees to trees, though given by a bimorphism over words. In fact, if we wish $e_1, e_2 \dots e_n$ to be also expressions, we would require a bimorphism over trees. This case is detailed as an example in Section~\ref{sec:illustratingExample}.

Notice that bimorphisms over words which do not produce placeholder variables anywhere already capture rational relations over words. Since equivalence-checking is undecidable for rational word relations~\cite{FR1968,Griffiths1968}, it is undecidable for our model as well. However we show that the \tcp\ is decidable, even for  transductions defined by bimorphisms over trees and potentially producing variables in the output of both morphisms. We first motivate the need for type-checking in the next paragraph, and give insights to our algorithm in the paragraphs following that. 

Continuing with our example from above, we may want to check that applying the transformation does not increase the alternation between the operators $\cdot$ and $+$. This can be checked by verifying that if the transformation is applied to any tree from a language $L$ consisting of trees with alternation depth bounded by a constant $c$, the resulting tree is in $L$ as well. On the other hand, for reasons of numerical stability, we may be interested in ensuring that size of subtrees made up of only the $+$ operator is bounded. It can be checked that this regular property is not preserved by the transformation. Instead of arithmetic expressions, if we consider Boolean expressions made up of $\land$ and $\lor$ operators, we can reason about expressions being in normal forms like conjunctive normal form (CNF), disjunctive normal form (DNF) etc. We consider the \tcp\ in general where the source trees have a property represented by a regular language $\sourceLanguage$ and output trees are required to have a property represented by another (possibly different) regular language $\targetLanguage$.

The main idea of our decision procedure for the \tcp\ is the following observation: if the given instance of the \tcp\ is a negative instance, a counterexample is a single tree in the seed language $\seedLanguage$, satisfying multiple conditions under multiple transformations. On the source side, applying the morphism $\sourceMap$ to the seed tree and then applying a substitution for the placeholder variables should result in a tree in $\sourceLanguage$. On the output side, applying the morphism $\targetMap$ to the same seed tree and then applying the same substitution for the placeholder variables should result in a tree not in $\targetLanguage$. These conditions can be checked by an alternating tree automaton running on trees in the seed language, thus reducing the \tcp\ to the emptiness-checking of alternating tree automata.

Along the way, we consider some related/simpler problems. The simple question of whether a given tree matches a given pattern is \nptc, a result known already for words~\cite{tata}. The special case of the \tcp\ where the transformation is given by a single pair of patterns is \exptc. This generalises the work of pattern-based word transformations~\cite{AMP2022} to trees.

\paragraph*{Related work}
Pattern-based word transformations were considered  in~\cite{AMP2022}, but these consisted only of a single pair of patterns. We extend this in two ways --- by having  transformations of trees, and by having a formalism that consists of potentially infinitely many pairs of patterns. 

There are many transducer models for trees. Our transformations subsume classical top-down and bottom-up tree transducers, even without any variables, thanks to the bimorphism feature. In fact it subsumes any relation on trees definable by bimorphisms~\cite[Chapter 6]{tata}. However it is incomparable to streaming tree transducers~\cite{StreamingTreeTransd} or macro tree transducers~\cite{MacroTreeTransd}. The domains of transformations defined by the latter are regular tree languages, which need not be the case in our model due to the presence of variables. On the other hand our model cannot recognize transformations which reverse a word (a tree with unary alphabet and an end marker as a leaf), while macro tree transducers and streaming tree transducers can. For exactly the same reasons, our model is incomparable with \texorpdfstring{$k$}{k}-Pebble Transducers~\cite{MSV2000}. See Appendix~\ref{sec:kPebbleTransducerComparison} for details.

A model of transition systems was studied in~\cite{GMSZ2008} where states are unranked, unordered forests and edges are given by tree rewriting using a pattern-based query language. These patterns are weaker than our patterns since a variable can represent only a single node in the tree and not a subtree like in our patterns. Reachability in this transition system is undecidable, but if the depth of every forest in the transition system is bounded it becomes decidable though nonelementary. Our type-checking problem is not addressing reachability in a transition system, but only a one-step rewriting.   While the decidability result in~\cite{GMSZ2008} is limited to bounded depth trees, our model can handle unrestricted trees. Further, all our algorithms are of elementary complexity. 

\paragraph*{Paper organisation}
In Section~\ref{sec:illustratingExample}, we illustrate the correspondence between the artefacts of our model and the parts of trees and transformations that the model represents. We recall some preliminaries about trees, contexts, nondeterministic tree automata and regular tree languages in Section~\ref{subsec:preliminaries} and formally define the model in Section~\ref{subsec:model}. In Section~\ref{sec:motiveEx}, we give some examples of possible applications, which motivate our model as well as the \tcp\ that we consider, which is formally defined in Section~\ref{sec:problemsAndResults}. We give some intuition underlying the decision procedure in Section~\ref{sec:proofOutline}, motivating the concept of alternating tree automata and relational lifts that are formally defined in Section~\ref{sec:relationalLifts}. The proof of decidability is put together in Section~\ref{sec:proofDetails}. Some of the detailed proofs are moved to the appendix due to space constraints.


\section{Illustrating Example}\label{sec:illustratingExample}

As mentioned in the introduction, our objective of this section is to illustrate our model through an example. We want to capture transformations of the form depicted in Figure~\ref{fig:trans-ex}. We will define a suitable seed language and give bimorphisms that attain this. 

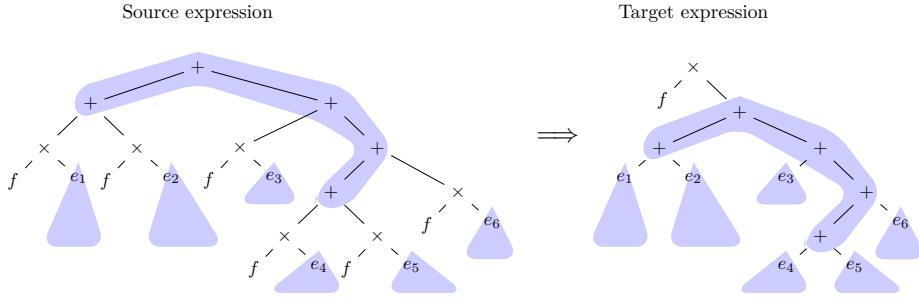
\begin{figure}[!ht]
    \scalebox{0.7}{\begin{tikzpicture}[node distance=5mm]
        \node[] (p1) {+};
        \node[below right=2mm and 20mm of p1] (p2) {+};
        \node[below left=2mm and 15mm of p1] (p3) {+};
        \node[below right= of p2] (p4) {+};
        \node[below left= of p4] (p5) {+};
        
        \node[below left= of p3] (m1) {$\times$};
        \node[below right= of p3] (m2) {$\times$};
        \node[below left=3.5mm and 12mm of p2] (m3) {$\times$};
        \node[below left= of p5] (m4) {$\times$};
        \node[below right= of p5] (m5) {$\times$};
        \node[below right=3.5mm and 10mm of p4] (m6) {$\times$};
        
        \node[below left=1mm and 1mm of m1] (t1) {$f$};
        \node[below left=1mm and 1mm of m2] (t2) {$f$};
        \node[below left=1mm and 1mm of m3] (t3) {$f$};
        \node[below left=1mm and 1mm of m4] (t4) {$f$};
        \node[below left=1mm and 1mm of m5] (t5) {$f$};
        \node[below left=1mm and 1mm of m6] (t6) {$f$};
        \node[below right=1mm and 1mm of m1] (e1) {$e_1$};
        \node[below right=1mm and 1mm of m2] (e2) {$e_2$};
        \node[below right=1mm and 1mm of m3] (e3) {$e_3$};
        \node[below right=1mm and 1mm of m4] (e4) {$e_4$};
        \node[below right=1mm and 1mm of m5] (e5) {$e_5$};
        \node[below right=1mm and 1mm of m6] (e6) {$e_6$};

        \draw[-]
        (p1) edge (p2)
        (p1) edge (p3)
        (p2) edge (p4)
        (p4) edge (p5);
    
        \draw[-]
        (p3) edge (m1)
        (p3) edge (m2)
        (p2) edge (m3)
        (p5) edge (m4)
        (p5) edge (m5)
        (p4) edge (m6);
    
        \draw[-]
        (m1) edge (t1)
        (m1) edge (e1)
        (m2) edge (t2)
        (m2) edge (e2)
        (m3) edge (t3)
        (m3) edge (e3)
        (m4) edge (t4)
        (m4) edge (e4)
        (m5) edge (t5)
        (m5) edge (e5)
        (m6) edge (t6)
        (m6) edge (e6);
    
        \node[above= of p1] (text){{\large Source expression}};	
        
        \begin{scope}[on background layer={}]
            \fill[rounded corners=3mm, color=blue!20, distance=5mm] ($(p1)+(0,+3mm)$)--($(p3)+(-3mm,+3mm)$)--($(p3)+(-3mm,-3mm)$)--($(p1)+(0,-3mm)$)--($(p2) +(-3mm, -3mm)$)--($(p4)+(-3mm,0)$)--($(p5) +(-3mm, 0mm)$)--($(p5) + (-3mm, -3mm)$)--($(p5) + (3mm, -3mm)$)--($(p4)+(3mm,0)$)--($(p2)+(3mm,3mm)$)--($(p1)+(0,+3mm)$);
        \end{scope}
        \begin{scope}[on background layer={}]
            \fill[rounded corners=3mm, color=blue!20, distance=5mm] ($(e1)+(0,+3mm)$)--($(e1)+(-7mm,-13mm)$) --($(e1)+(5mm,-13mm)$)--($(e1)+(0,+3mm)$) ;
        \end{scope}
        \begin{scope}[on background layer={}]
            \fill[rounded corners=3mm, color=blue!20, distance=5mm] ($(e2)+(0,+3mm)$)--($(e2)+(-5mm,-13mm)$) --($(e2)+(10mm,-13mm)$)--($(e2)+(0,+3mm)$) ;
        \end{scope}
        \begin{scope}[on background layer={}]
            \fill[rounded corners=3mm, color=blue!20, distance=5mm] ($(e3)+(0,+3mm)$)--($(e3)+(-7mm,-5mm)$) --($(e3)+(5mm,-5mm)$)--($(e3)+(0,+3mm)$) ;
        \end{scope}
        \begin{scope}[on background layer={}]
            \fill[rounded corners=3mm, color=blue!20, distance=5mm] ($(e4)+(0,+3mm)$)--($(e4)+(-10mm,-5mm)$) --($(e4)+(5mm,-5mm)$)--($(e4)+(0,+3mm)$) ;
        \end{scope}
        \begin{scope}[on background layer={}]
            \fill[rounded corners=3mm, color=blue!20, distance=5mm] ($(e5)+(0,+3mm)$)--($(e5)+(-5mm,-5mm)$) --($(e5)+(10mm,-5mm)$)--($(e5)+(0,+3mm)$) ;
        \end{scope}
        \begin{scope}[on background layer={}]
            \fill[rounded corners=3mm, color=blue!20, distance=5mm] ($(e6)+(0,+3mm)$)--($(e6)+(-6mm,-7mm)$) --($(e6)+(5mm,-7mm)$)--($(e6)+(0,+3mm)$) ;
        \end{scope}
    \end{tikzpicture}}
    \raisebox{2cm}{ $\Longrightarrow$}
    \scalebox{0.7}{\begin{tikzpicture}[node distance=5mm]
        \node[] (p1) {+};
        \node[below right=2mm and 10mm of p1] (p2) {+};
        \node[below left=2mm and 10mm of p1] (p3) {+};
        \node[below right= of p2] (p4) {+};
        \node[below left= of p4] (p5) {+};
        \node[above left= of p1] (m) {$\times$};
        \node[below left=1mm and 1mm of m] (t) {$f$};
        \node[below left= 1mm and 1mm of p3] (e1) {$e_1$};
        \node[below right=1mm and 1mm of p3] (e2) {$e_2$};
        \node[below left=1mm and 1mm of p2] (e3) {$e_3$};
        \node[below left=1mm and 1mm of p5] (e4) {$e_4$};
        \node[below right=1mm and 1mm of p5] (e5) {$e_5$};
        \node[below right=1mm and 1mm of p4] (e6) {$e_6$};

        \draw[-]
        (p1) edge (p2)
        (p1) edge (p3)
        (p2) edge (p4)
        (p4) edge (p5);
        
        \draw[-]
        (p3) edge (e1)
        (p3) edge (e2)
        (p2) edge (e3)
        (p5) edge (e4)
        (p5) edge (e5)
        (p4) edge (e6);
        
        \draw[-]
        (m) edge (t)
        (m) edge (p1);

        \node[above= of m] (text){{\large Target expression}};	
                                    
        \begin{scope}[on background layer={}]
            \fill[rounded corners=3mm, color=blue!20, distance=5mm] ($(p1)+(0,+3mm)$)--($(p3)+(-3mm,+3mm)$) --($(p3)+(-3mm,-3mm)$)-- ($(p1)+(0,-3mm)$)--($(p2) +(-3mm, -3mm)$) --($(p4)+(-3mm,0)$) -- ($(p5) +(-3mm, 0mm)$) -- ($(p5) +(-3mm, -3mm)$) --($(p5) +(3mm, -3mm)$) --($(p4)+(3mm,0)$) --($(p2)+(3mm,3mm)$)--($(p1)+(0,+3mm)$);
        \end{scope}
        \begin{scope}[on background layer={}]
            \fill[rounded corners = 3mm, color=blue!20, distance=5mm] ($(e1)+(0,+3mm)$)--($(e1)+(-7mm,-13mm)$) --($(e1)+(5mm,-13mm)$)--($(e1)+(0,+3mm)$) ;
        \end{scope}
        \begin{scope}[on background layer={}]
            \fill[rounded corners = 3mm, color=blue!20, distance=5mm] ($(e2)+(0,+3mm)$)--($(e2)+(-5mm,-13mm)$) --($(e2)+(10mm,-13mm)$)--($(e2)+(0,+3mm)$) ;
        \end{scope}
        \begin{scope}[on background layer={}]
            \fill[rounded corners = 3mm, color=blue!20, distance=5mm] ($(e3)+(0,+3mm)$)--($(e3)+(-7mm,-5mm)$) --($(e3)+(5mm,-5mm)$)--($(e3)+(0,+3mm)$) ;
        \end{scope}
        \begin{scope}[on background layer={}]
            \fill[rounded corners = 3mm, color=blue!20, distance=5mm] ($(e4)+(0,+3mm)$)--($(e4)+(-10mm,-5mm)$) --($(e4)+(5mm,-5mm)$)--($(e4)+(0,+3mm)$) ;
        \end{scope}
        \begin{scope}[on background layer={}]
            \fill[rounded corners = 3mm, color=blue!20, distance=5mm] ($(e5)+(0,+3mm)$)--($(e5)+(-5mm,-5mm)$) --($(e5)+(10mm,-5mm)$)--($(e5)+(0,+3mm)$) ;
        \end{scope}
        \begin{scope}[on background layer={}]
            \fill[rounded corners = 3mm, color=blue!20, distance=5mm] ($(e6)+(0,+3mm)$)--($(e6)+(-6mm,-7mm)$) --($(e6)+(5mm,-7mm)$)--($(e6)+(0,+3mm)$) ;
        \end{scope}
    \end{tikzpicture}}
    \caption{The transformation of an expression following distributivity. Here $f$, $e_1$, \dots $e_6$ are arbitrary expressions. If $f$ is a variable instead of an expression, we get a (source pattern, target pattern) pair that captures this transformation.}\label{fig:trans-ex}
\end{figure} 

\newcommand{\insertcontext}{\scalebox{.8}{
    \begin{tikzpicture}[node distance=1mm, inner sep=0.5mm]
    \node[] (p1) {};
    \node[above left= of p1] (m) {$\times$};
    \node[below left= of m] (t) {$y$};
    
    \draw[-]
    (m) edge (t)
    (m) edge (p1);		
\end{tikzpicture}}}

Consider the transformation depicted in Figure~\ref{fig:trans-ex}. The blue shaded parts in the source expression and the target expression remain unchanged. Expressions are formed by the alphabet $\set{+, \times, 0, 1}$. The patterns would require a placeholder symbol $y$ in addition to the above alphabet. In order to generate the (source pattern, target pattern) pair capturing this transformation, we require a seed tree which carries this blue shaded part intact. We expect the morphisms $\sourceMap$ and $\targetMap$ to be identity on this part. However, $\insertcontext$ gets inserted above $e_i$ in the source pattern, and at the root in the target pattern. In order to handle this we add new unary symbols $a$ and $b$ in the seed tree.

The seed tree and the morphisms $\sourceMap, \targetMap$ for this transformation are depicted in Figure~\ref{fig:seed-tree}.

\begin{figure}[!ht]
    \begin{tikzpicture}[node distance=5mm]
        \node[] (p1) {+};
        \node[below right=2mm and 10mm of p1] (p2) {+};
        \node[below left=2mm and 10mm of p1] (p3) {+};
        \node[below right= of p2] (p4) {+};
        \node[below left= of p4] (p5) {+};
        \node[above=2mm of p1] (m) {$a$};
        \node[below left=1mm and 1mm of p3] (b1) {$b$};
        \node[below right=1mm and 1mm of p3] (b2) {$b$};
        \node[below left=1mm and 5mm of p2] (b3) {$b$};
        \node[below left=1mm and 1mm of p5] (b4) {$b$};
        \node[below right=1mm and 1mm of p5] (b5) {$b$};
        \node[below right=1mm and 1mm of p4] (b6) {$b$};
        \node[below=2mm of b1] (e1) {$e_1$};
        \node[below=2mm of b2] (e2) {$e_2$};
        \node[below=2mm of b3] (e3) {$e_3$};
        \node[below=2mm of b4] (e4) {$e_4$};
        \node[below=2mm of b5] (e5) {$e_5$};
        \node[below=2mm of b6] (e6) {$e_6$};

        \draw[-]
        (p1) edge (p2)
        (p1) edge (p3)
        (p2) edge (p4)
        (p4) edge (p5);
        
        \draw[-]
        (p3) edge (b1)
        (p3) edge (b2)
        (p2) edge (b3)
        (p5) edge (b4)
        (p5) edge (b5)
        (p4) edge (b6);
        
        \draw[-]
        (b1) edge (e1)
        (b2) edge (e2)
        (b3) edge (e3)
        (b4) edge (e4)
        (b5) edge (e5)
        (b6) edge (e6);
        
        \draw[-]
        (m) edge (p1);
        
        \node[left= of m] (text){{Seed tree}};
            
        \begin{scope}[on background layer={}]
            \fill[rounded corners=3mm, color=blue!20, distance=5mm] ($(p1)+(0,+3mm)$)--($(p3)+(-3mm,+3mm)$)--($(p3)+(-3mm,-3mm)$)--($(p1)+(0,-3mm)$)--($(p2) +(-3mm, -3mm)$) --($(p4)+(-3mm,0)$) -- ($(p5) +(-3mm, 0mm)$) -- ($(p5) +(-3mm, -3mm)$) --($(p5) +(3mm, -3mm)$) --($(p4)+(3mm,0)$) --($(p2)+(3mm,3mm)$)--($(p1)+(0,+3mm)$);
        \end{scope}
        \begin{scope}[on background layer={}]
            \fill[rounded corners=3mm, color=blue!20, distance=5mm] ($(e1)+(0,+3mm)$)--($(e1)+(-7mm,-13mm)$) --($(e1)+(5mm,-13mm)$)--($(e1)+(0,+3mm)$);
        \end{scope}
        \begin{scope}[on background layer={}]
            \fill[rounded corners=3mm, color=blue!20, distance=5mm] ($(e2)+(0,+3mm)$)--($(e2)+(-5mm,-13mm)$) --($(e2)+(10mm,-13mm)$)--($(e2)+(0,+3mm)$);
        \end{scope}
        \begin{scope}[on background layer={}]
            \fill[rounded corners=3mm, color=blue!20, distance=5mm] ($(e3)+(0,+3mm)$)--($(e3)+(-7mm,-5mm)$) --($(e3)+(5mm,-5mm)$)--($(e3)+(0,+3mm)$);
        \end{scope}
        \begin{scope}[on background layer={}]
            \fill[rounded corners=3mm, color=blue!20, distance=5mm] ($(e4)+(0,+3mm)$)--($(e4)+(-10mm,-5mm)$) --($(e4)+(5mm,-5mm)$)--($(e4)+(0,+3mm)$);
        \end{scope}
        \begin{scope}[on background layer={}]
            \fill[rounded corners=3mm, color=blue!20, distance=5mm] ($(e5)+(0,+3mm)$)--($(e5)+(-5mm,-5mm)$) --($(e5)+(10mm,-5mm)$)--($(e5)+(0,+3mm)$);
        \end{scope}
        \begin{scope}[on background layer={}]
            \fill[rounded corners=3mm, color=blue!20, distance=5mm] ($(e6)+(0,+3mm)$)--($(e6)+(-6mm,-7mm)$) --($(e6)+(5mm,-7mm)$)--($(e6)+(0,+3mm)$);
        \end{scope}
    \end{tikzpicture}
    \hfil
    \raisebox{2.3cm}{\framebox[0.45\textwidth]{\begin{tabularx}{0.37\textwidth}{c c c c c}
        \multicolumn{5}{c}{Morphisms}\\\\
        $\e$ & $\xleftarrow{\sourceMap}$ & $a$ & $\xrightarrow{\targetMap} $ & $\insertcontext$\\
        $\insertcontext$ & $\xleftarrow{\sourceMap}$ & $b$ & $\xrightarrow{\targetMap} $ & $\e$\\
        $+$ & $\xleftarrow{\sourceMap}$ & $+$ & $\xrightarrow{\targetMap} $ & $+$\\
        $\times$ & $\xleftarrow{\sourceMap}$ & $\times$ & $\xrightarrow{\targetMap} $ & $\times$\\
        $0$ & $\xleftarrow{\sourceMap}$ & $0$ & $\xrightarrow{\targetMap} $ & $0$\\
        $1$ & $\xleftarrow{\sourceMap}$ & $1$ & $\xrightarrow{\targetMap} $ & $1$\\								
    \end{tabularx}}}
    \caption{The seed tree and the morphisms generating the (source pattern, target pattern) pair capturing the transformation depicted in Figure~\ref{fig:trans-ex}.}\label{fig:seed-tree}
\end{figure}
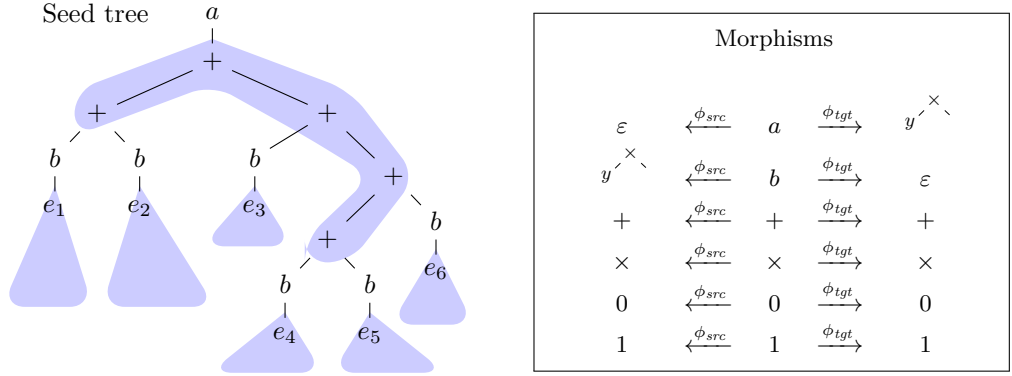

Note that, we require one such seed tree for each possible choice of $e_i$ and the shape of the ``+-kernel'', thus giving rise to an infinite seed language and hence an infinite set of (source pattern, target pattern) pairs capturing the transformation. Thankfully, the set of seed trees is tree-regular (need to check that the root is an $a$, every path has exactly one $b$, and the segment between $a$ and $b$ has only $+$s).

We now give an example of \tcp\ on this transformation. We define the product-depth of an expression as follows. The product-depth (pd) of a branch is the number of occurrences of $\times$ in that branch. The product-depth of an expression is the maximum of the product-depths of its branches. The classes of $k$-pd-bounded expressions are those with product-depth at most $k$. We  may be interested in checking whether $k$-pd-boundedness is preserved by the above transformation. This is an instance of \tcp, since $k$-pd-bounded expressions form a regular set of trees.

Further examples, in particular the ones where we need multiple variables, are presented after formally defining our model, which is done next.


\section{Formal Model}\label{sec:formalModel}
\subsection{Preliminaries}\label{subsec:preliminaries}
\paragraph*{Sets and Relations}\label{par:setAndRelations}

Let $\N := \set{1, 2, 3, \dots}$ be the set of strictly positive integers and let $\N_0 := \N \cup \set{0}$. For $n \in \N$, we denote the set $\set{1, \ldots, n}$ by $[n]$. As a convention, we let $[0] := \emptyset$. For a set $X$ and $n \in \N$, $X^n$ denotes the usual Cartesian power of $X$, consisting of all $n$-tuples of elements in $X$. As a convention, we let $X^0 := \{\natural\}$, a special singleton set.

For $n \in \N_0$, an $n$-ary relation over $X$ is a subset $R \subseteq X^n$. We denote the set of all $n$-ary relations on $X$ by $\relations(X, n)$. Finally, we denote $\bigcup_{k=0}^\infty \relations(X, k)$ by $\relations(X)$.

We consider automata over finite, ranked trees, following definitions from~\cite{tata}.

\begin{definition}[Alphabet]\label{def:alphabet}
    A finite set $\alphabet := \set{a_1/k_1, \dots, a_n/k_n}$ is called a \emph{finite ranked alphabet} or \emph{alphabet}, where $a_i$ are called the \emph{letters}, and $k_i =: \arity(a_i)$ are nonnegative integers called their corresponding \emph{arities} for all $i \in [n]$. We define pairwise disjoint sets $\alphabet_k := \set{a \in \alphabet \mid \arity(a) = k} \subseteq \alphabet$ for all $k \in \N_0$. We require that $\alphabet_0 \neq \emptyset$. We also define the arity of $\alphabet := \max_i k_i$.
\end{definition}

\begin{definition}[Term, Tree]\label{def:tree}
    A \emph{term} or a \emph{tree} over a ranked alphabet $\alphabet$ is a partial function $t : \N^* \rightharpoonup \alphabet$ with the domain $\Positions{t}$ satisfying:
    \begin{itemize}
        \item $\Positions{t}$ is nonempty and prefix-closed.
        \item $\forall u \in \Positions{t}, t(u) \in \alphabet_k, k \in \N_0 \implies \set{j \mid uj \in \Positions{t}} = [k]$.
    \end{itemize}
    We say that a tree $t$ is \emph{finite} if $\Positions{t}$ is finite.
\end{definition}
In this paper, we will only deal with finite trees. We denote the set of all finite trees over $\alphabet$ by $\trees{\alphabet}$. A \emph{tree language} is any subset of $\trees{\alphabet}$. Regular tree languages are those recognised by tree automata as defined below.

\begin{definition}[Nondeterministic Bottom-up Tree \{Transition System, Automaton\}]\label{def:nbtts}
    Let $\alphabet$ be a finite alphabet. A \emph{Nondeterministic Bottom-up Tree Transition System} or \emph{NBTTS} over $\alphabet$ is defined as $\transitionSystem := (\automatonStateSet, \alphabet, \automatonTransitionsSet)$, where $\automatonStateSet$ is a finite set of \emph{states}, and $\automatonTransitionsSet : \alphabet \to \relations(\automatonStateSet)$ is a function such that $\Transitions(a/k) \in \relations(\automatonStateSet, k+1)$, for all $a/k \in \alphabet$. A \emph{Nondeterministic Bottom-up Tree Automaton} or an \emph{NBTA} is a tuple $\automaton := (\automatonStateSet, \alphabet, \automatonTransitionsSet, \automatonFinalStatesSet)$, where $(\automatonStateSet, \alphabet, \automatonTransitionsSet)$ is a NBTTS over $\alphabet$ and $\automatonFinalStatesSet \subseteq \automatonStateSet$ is a set of \emph{final} states. If for each $a/k \in \alphabet$ we have that $\automatonTransitionsSet(a/k)$ is such that for all $q_1, \dots, q_k \in \automatonStateSet$, there exists a unique $q \in \automatonStateSet$ such that $(q_1, \dots, q_k, q) \in \automatonTransitionsSet(a)$, then the corresponding NBTTS is called a \emph{Deterministic Bottom-up Tree Transition System} or \emph{DBTTS}, and the NBTA is called a \emph{Deterministic Bottom-up Tree Automaton} or \emph{DBTA}.
\end{definition}

\begin{definition}[\{Run, Result, Acceptance\} for NBTA]\label{def:acceptanceByNbta}
    Given an NBTTS $\transitionSystem := (\automatonStateSet, \alphabet, \automatonTransitionsSet)$ and a tree $t \in \trees{\alphabet}$, a \emph{run} of $\transitionSystem$ on $t$ is defined as a tree $\run : \Positions{t} \to \automatonStateSet$ which is `compatible' with $\automatonTransitionsSet$, i.e., for every position $u \in \Positions{t}$, if $t(u) = a/k, \run(u) = q$ and $\run(ui) := q_i$ for all $i \in [k]$, then, $(q_1, \dots, q_k, q) \in \automatonTransitionsSet(a)$. The \emph{result} of a run $\run$ is $\result(\run) := \run(\e)$. Given an NBTA $\automaton := (\automatonStateSet, \alphabet, \automatonTransitionsSet, \automatonFinalStatesSet)$, and a tree $t \in \trees{\alphabet}$, we say that $t$ is \emph{accepted} by $\automaton$ if and only if there exists a run $\run$ of $(\automatonStateSet, \alphabet, \automatonTransitionsSet)$ on $t$ such that $\result(\run) \in \automatonFinalStatesSet$. The set of all trees $t$ accepted by a NBTA $\automaton$ is called the language \emph{recognised} by $\automaton$, and is denoted by $\languageOf{\automaton} \subseteq \trees{\alphabet}$.
\end{definition}

Just like a word $w$ can be extended by appending another word to its end, a tree can be extended by appending other trees to designated leaves. This is formalised below, where $ \uplus$ denotes disjoint union.

\begin{definition}[Context]\label{def:context}
    A \emph{context with arity $n$} or an \emph{$n$-context} over $\alphabet$ is a tree over the alphabet $\alphabet \uplus \set{\square_1/0, \dots, \square_n/0}$ such that each letter $\square_i$ appears exactly once in the tree, for all $i \in [n]$. We refer to the positions occupied by $\square_i$ as the \emph{holes} of the context for all $i \in [n]$.
\end{definition}

In this paper, we shall only deal with finite contexts. Intuitively, the leaves labeled with $\square_1, \dots, \square_n$, or the holes of the context are the designated positions where other trees can be appended. We denote the set of all finite $n$-contexts over $\alphabet$ by $\contexts{\alphabet}{n}$. We denote $\bigcup_{n=0}^\infty \contexts{\alphabet}{n}$ by $\contexts{\alphabet}{}$, where we let $\contexts{\alphabet}{0} := \trees{\alphabet}$.

\subsection{Tree Transformations}\label{subsec:model}

In this section, we define our tree transformation model and describe some associated problems that we wish to study. We start by defining and describing tree homomorphisms, which are essentially extensions of homomorphisms from words to trees. We more or less follow the definition as presented in~\cite{tata}, with minor modifications for ease of writing and mathematical rigour.

\begin{definition}[Tree Homomorphism]\label{def:treeHomomorphism}
    Let $\Sigma, \Gamma$ be two finite ranked alphabets, possibly not disjoint. Let $X := \set{x_i/0 \mid i \in \N}$ be a countable set of \emph{variables} of arity zero that is pairwise disjoint from $\Sigma, \Gamma$. For each $n \in \N_0$, define the \emph{section} of $X$ by $n$ by $X_n := \set{x_i/0 \mid i \in [n]}$. Note that with this definition, we also get that $X_0 = \emptyset$.

    Let $h : \Sigma \to \trees{\Gamma \uplus X}$ be a mapping such that $h(f) \in \trees{\Gamma \uplus X_{\arity(f)}}$ for all $f \in \Sigma$. Then, we define the \emph{tree homomorphism} $\widetilde{h} : \trees{\Sigma} \to \trees{\Gamma}$ determined by $h$ as follows:
    \begin{itemize}
        \item $\widetilde{h}(a) := h(a) \in \trees{\Gamma}$ for each $a \in \Sigma$ with arity 0.
        \item $\widetilde{h}(f(t_1, \dots, t_n)) := h(f)\{x_i \leftarrow \widetilde{h}(t_i) \mid i \in [n]\}$ for each $f/n \in \Sigma$, where the right-hand side is the result of replacing all occurences of the symbol $x_i$ in $h(f)$ by the term $\widetilde{h}(t_i)$ for all $i \in [n]$.
    \end{itemize}
\end{definition}

By abuse of notation, whenever there is no ambiguity, we shall drop the tilde and just refer to the tree homomorphism $\widetilde{h}$ determined by $h$ just as simply $h$.

We let $\alphabet$ be a finite ranked alphabet. We will need trees, some parts of which are not completely specified and we use variables to denote those parts. Let $\variables$ be a countable set of \emph{ranked variables} or \emph{variables}. For a variable $x \in \variables$, $\arity(x) \in \N_0$ denotes its arity. We define pairwise disjoint sets $\variables_k := \{x \in \variables \mid \arity(x) = k\} \subseteq \variables$ for all $k \in \N$. Further we also define the \emph{degree} of a variable $x$ as $\deg(x) := 1 + \arity(x)$.

\begin{definition}[Pattern]\label{def:pattern}
    A \emph{tree pattern} or simply a \emph{pattern} over $(\alphabet, \variables)$ is defined to be a tree $\pattern \in \trees{\alphabet \uplus \variables}$. We also let \emph{$\variables(\pattern)$} be the set of all variables appearing in $\pattern$. Note that $\variables(\pattern) \subseteq \variables$ is necessarily finite. We also define the \emph{degree of a pattern} $\pattern$ denoted by $\deg(\pattern) := \max_{x \in \variables(\pattern)}\deg(x)$.
\end{definition}

A pattern can represent multiple trees by substituting variables in the pattern by contexts of the same arity.
\begin{definition}[Substitution]\label{def:substitution}
    A \emph{substitution} over $\alphabet$ is a partial function $\substitution : \variables \rightharpoonup \contexts{\alphabet}{}$ such that for all $x \in \dom(\substitution)$, we have that $\substitution(x) \in \contexts{\alphabet}{k}$, where $k = \arity(x)$.
\end{definition}

Given a substitution $\substitution$, we can naturally extend it to the set of patterns over $\alphabet$ as follows. Define a tree homomorphism $\widetilde{\substitution} : \trees{\alphabet \uplus \dom(\substitution)} \to \trees{\alphabet}$ by setting
\[\twopartfunc{\widetilde{\substitution}(x)}{x}{x \in \alphabet}{\substitution(x)}{x \in \dom(\substitution)}\]
and extending it homomorphically. Since $\widetilde{\substitution}|_{\dom(\substitution)} = \substitution$, by abuse of notation we shall write $\widetilde{\substitution}$ as $\substitution$ in places where the distinction is irrelevant.

We would like to control the set of contexts that can be substituted for variables in patterns, for which we use guards. A \emph{guard} over a finite alphabet $\alphabet$ is a partial function $\guardFunction : \variables \to 2^{\trees{\alphabet}}$, which assigns to each variable in its domain, a tree language over $\alphabet$. The following formally defines when a pattern matches a tree.

\begin{definition}[Match]\label{def:match}
    Let $\pattern$ be a pattern over $(\alphabet, \variables)$ and $t \in \trees{\alphabet}$ be a tree. Let $\guardFunction$ be a guard such that $\variables(\pattern) \subseteq \dom(\guardFunction)$. We say that $t$ \emph{matches} pattern $\pattern$ under $\guardFunction$ if and only if there exists a substitution $\substitution$ such that
    \begin{itemize}
        \item $\variables(\pattern) \subseteq \dom(\substitution) \subseteq \dom(\guardFunction)$.
        \item $\substitution(x) \in \guardOf{x}$ for each $x \in \variables(\pattern)$.
        \item $\substitution(\pattern) = t$.
    \end{itemize}
    In such a case, $\substitution$ is called a \emph{matching substitution}, and we say the substitution $\substitution$ \emph{matches} the tree $t$ to the pattern $\pattern$ under the guard $\guardFunction$.
\end{definition}

There may exist more than one matching substitution for a given match. Eg., if $a/1, b/1, c/0$ are letters and $x/1, y/0$ are variables, the tree $a-b-b-c$ matches the pattern $x-b-y$ with 2 distinct substitutions $\substitution_1, \substitution_2$ given by $\substitution_1(x) := a, \substitution_1(y) := b-c$ whereas $\substitution_2(x) := a-b, \substitution_2(y) := c$.

The following formalises the tree transformations defined by our model.
\begin{definition}[Atomic Transformation]\label{def:atomicAction}
    Let $\alphabet$ be a finite alphabet and $\variables$ be a countable set of variables. An \emph{atomic tree transformation} or simply, an \emph{atomic transformation} is an ordered triple $\atomicTransform := (\sourcePattern, \targetPattern, \guardFunction)$, where $\sourcePattern, \targetPattern$ are patterns over $(\alphabet, \variables)$ called the \emph{source pattern} and \emph{target pattern} respectively, and $\guardFunction$ is a guard over $\alphabet$ such that we have $\variables(\atomicTransform) := \variables(\sourcePattern) \cup \variables(\targetPattern) \subseteq \dom(\guardFunction)$. The \emph{relation over trees induced by \atomicTransform} is
    \begin{eqnarray*}
        \relT(\atomicTransform)
        & := & \{(t,t') \mid \exists \substitution \mbox{ such that } \substitution \mbox{ matches } t  \mbox{ to } \sourcePattern \mbox{ under } \guardFunction\\
        &  & \mbox{ and } \substitution \mbox{ matches } t'  \mbox{ to } \targetPattern \mbox{ under } \guardFunction \}
    \end{eqnarray*}
\end{definition}
We denote by $\atomicTransform(t) = \{t' \mid (t,t') \in \relT(\atomicTransform)\} \subseteq \trees{\alphabet}$ the \emph{effect} of $\atomicTransform$ on the tree $t$. We naturally extend the effect of $\atomicTransform$ to tree languages $\genericLanguage$ by setting $\atomicTransform(\genericLanguage) := \bigcup_{t \in \genericLanguage}\atomicTransform(t)$. We also define $\deg(\atomicTransform) := \max_{x \in \variables(\atomicTransform)}\deg(x)$. Our model consists of a possibly infinite set of atomic transformations, represented by a bimorphism.

\begin{definition}[Transformation]\label{def:action}
    Let $\alphabet, \seedAlphabet$ be two finite alphabets. Let $\variables$ be a countably infinite set of variables. A \emph{transformation} $\fullTransform$ is a tuple $\fullTransform := (\seedLanguage, \sourceMap, \targetMap, \guardFunction)$ where
    \begin{itemize}
        \item $\seedLanguage \subseteq \trees{\seedAlphabet}$ is called the \emph{seed language}.
        \item $\sourceMap, \targetMap : \trees{\seedAlphabet} \to \trees{\alphabet \uplus \variables}$ are tree homomorphisms.
        \item $\guardFunction$ is a guard over $\alphabet$.
    \end{itemize}
    We define the relation induced by a transformation $\fullTransform := (\seedLanguage, \sourceMap, \targetMap, \guardFunction)$ as follows: for each seed tree $s \in \seedLanguage$, we define the atomic transformation $\atomicTransform_s:= (\sourceMap(s), \targetMap(s), \guardFunction)$. Then $\relT(\fullTransform) := \bigcup_{s \in \seedLanguage}\relT(\atomicTransform_s)$
\end{definition}
For a tree $t$ and a tree language $\genericLanguage$, $\fullTransform(t) = \{t' \mid (t,t') \in \relT(\fullTransform)\}\subseteq \trees{\alphabet}$ and $\fullTransform(\genericLanguage) := \bigcup_{t \in \genericLanguage}\fullTransform(t)$ as before. We further define $\variables(\fullTransform)= \bigcup_{\gamma \in \seedAlphabet} \variables(\sourceMap(\gamma)) \cup \variables(\targetMap(\gamma))$ and $\deg(\fullTransform) = \max_{x \in \variables(\fullTransform)}\deg(x)$.


\section{Example - Code Snippet Translation}\label{sec:forWhileLoopExample}\label{sec:motiveEx}

As a motivating example, we consider a piece of code, written in a language like C++ or Java, containing a certain bounded number of \textbf{for} loops. We would like to transform this into a corresponding piece of code with \textbf{while} loops. A concrete example is depicted in Figure~\ref{fig:code-snippets}.
\begin{figure}[!ht]
    {\footnotesize\begin{tabular}{p{0.45\textwidth}}
        Source code \\ \hline
        \texttt{int[] A = new int[m];}\\
        \texttt{for(int i=0;i<A.length;i++) \{}\\
        \texttt{\quad A[i] = getInputFromUser();}\\
        \texttt{\}}\\
        \texttt{for(int j=0;j<A.length;j++) \{}\\
        \texttt{\quad int minIndex = j;}\\
        \texttt{\quad int minValue = A[minIndex];}\\
        \texttt{\quad for(int k=j;k<A.length;k++) \{}\\
        \texttt{\quad \quad if(A[k] < minValue) \{}\\
        \texttt{\quad \quad \quad minIndex = k;}\\
        \texttt{\quad \quad \quad minValue = A[k];}\\
        \texttt{\quad \quad \}}\\
        \texttt{\quad \}}\\
        \texttt{\quad int temp = A[j];}\\
        \texttt{\quad int A[j] = A[k];}\\
        \texttt{\quad int A[k] = temp;}\\
        \texttt{\}}\\
        \texttt{return A;}\\
        \\
        \\
    \end{tabular}
    \hfil
    \begin{tabular}{p{0.45\textwidth}}
        Target code \\ \hline
        \texttt{int[] A = new int[m];}\\
        \texttt{int i=0;}\\
        \texttt{while (i<A.length) \{}\\
        \texttt{\quad A[i] = getInputFromUser();}\\
        \texttt{\quad i++;}\\
        \texttt{\}}\\
        \texttt{int j=0;}\\
        \texttt{while (j<A.length) \{}\\
        \texttt{\quad int minIndex = j;}\\
        \texttt{\quad int minValue = A[minIndex];}\\
        \texttt{\quad int k=j;}\\
        \texttt{\quad while(k<A.length) \{}\\
        \texttt{\quad \quad if (A[k] < minValue) \{}\\
        \texttt{\quad \quad \quad minIndex = k;}\\
        \texttt{\quad \quad \quad minValue = A[k];}\\
        \texttt{\quad \quad \}}\\
        \texttt{\quad \quad k++;}\\
        \texttt{\quad \}}\\
        \texttt{\quad int temp = A[j];}\\
        \texttt{\quad int A[j] = A[k];}\\
        \texttt{\quad int A[k] = temp;}\\
        \texttt{\quad j++;}\\
        \texttt{\}}\\
        \texttt{return A;}\\
    \end{tabular}}
    \caption{The source and target code snippets for Selection Sort, before and after transforming \texttt{for} loops to \texttt{while} loops.}\label{fig:code-snippets}
\end{figure}

Let us bound the number of $\examplePatternLanguageFor$ loops in our programs by some fixed $n \in \N$. Consider the alphabet $\alphabet := \coreAlphabet \uplus \set{\examplePatternLanguageFor/5, \examplePatternLanguageWhile/3, \examplePatternLanguageNull/0}$ where the unary $\coreAlphabet$ refers to the alphabet used for programming, which for instance, could be all the characters appearing on a standard keyboard. With this alphabet, the source and target codes are seen as trees, as depicted in Figure~\ref{fig:source-code-snippet-as-tree}. Long chains of characters from $\coreAlphabet$ are written as mere sequences wrapped in a box for the sake of readability and space. 

\begin{figure}[!ht]
    \centering
    \begin{subfigure}[t]{0.48\textwidth}
        \centering
        \vtop{\begin{tikzpicture}[node distance=2mm]
            \node[draw, align=left] (blob1) {\footnotesize \texttt{int[] A = new int[m];}};

            \node[draw, align=left, below=2mm of blob1] (for1) {\texttt{for}};
            \node[draw, align=left, below left=2mm of for1] (init1) {\texttt{int i=0;}};
            \node[draw, align=left, below left=2mm and -10mm of init1] (null1) {\texttt{null}};
            \node[draw, align=left, below right=2mm and -8mm of null1] (term1) {\texttt{i<A.length;}};
            \node[draw, align=left, below left=2mm and -10mm of term1] (null2) {\texttt{null}};
            \node[draw, align=left, below right=2mm and -3mm of null2] (incr1) {\texttt{i++}};
            \node[draw, align=left, below left=2mm and -10mm of incr1] (null3) {\texttt{null}};
            \node[draw, align=left, below right=2mm and -25mm of null3] (blob2) {\texttt{A[i] = getInputFromUser();}};
            \node[draw, align=left, below left=2mm and -10mm of blob2] (null4) {\texttt{null}};

            \node[draw, align=left, below right=2mm and 19mm of null4] (for2) {\texttt{for}};
            \node[draw, align=left, below left=2mm of for2] (init2) {\texttt{int j=0;}};
            \node[draw, align=left, below left=2mm and -10mm of init2] (null5) {\texttt{null}};
            \node[draw, align=left, below right=2mm and -8mm of null5] (term2) {\texttt{j<A.length;}};
            \node[draw, align=left, below left=2mm and -10mm of term2] (null6) {\texttt{null}};
            \node[draw, align=left, below right=2mm and -3mm of null6] (incr2) {\texttt{j++}};
            \node[draw, align=left, below left=2mm and -10mm of incr2] (null7) {\texttt{null}};
            \node[draw, align=left, below right=2mm and -14mm of null7] (blob3) {
                \texttt{int minIndex = j;}\\
                \texttt{int minValue = A[minIndex];}};
            \node[draw, align=left, below=2mm of blob3] (for3) {\texttt{for}};
            \node[draw, align=left, below left=2mm of for3] (init3) {\texttt{int k=j;}};
            \node[draw, align=left, below left=2mm and -10mm of init3] (null8) {\texttt{null}};
            \node[draw, align=left, below right=2mm and -8mm of null8] (term3) {\texttt{k<A.length;}};
            \node[draw, align=left, below left=2mm and -10mm of term3] (null9) {\texttt{null}};
            \node[draw, align=left, below right=2mm and -3mm of null9] (incr3) {\texttt{k++}};
            \node[draw, align=left, below left=2mm and -10mm of incr3] (null10) {\texttt{null}};
            \node[draw, align=left, below right=2mm and -14mm of null10] (blob4) {
                \texttt{if(A[k] < minValue) \{}\\
                \texttt{\quad minIndex = k;}\\
                \texttt{\quad minValue = A[k];}\\
                \texttt{\}}};
            \node[draw, align=left, below left=2mm and -10mm of blob4] (null11) {\texttt{null}};
            \node[draw, align=left, below right=2mm and 6mm of null11] (blob5) {
                \texttt{int temp = A[j];}\\
                \texttt{int A[j] = A[k];}\\
                \texttt{int A[k] = temp;}};
            \node[draw, align=left, below left=2mm and -10mm of blob5] (null12) {\texttt{null}};
            \node[draw, align=left, below right=2mm and 6mm of null12] (blob6) {\texttt{return A;}};
            \node[draw, align=left, below left=2mm and -10mm of blob6] (null13) {\texttt{null}};
    
            \draw[->]
            (blob1) edge (for1)
            
            (for1) edge (init1)
            (init1) edge (null1)
            (for1) edge (term1)
            (term1) edge (null2)
            (for1) edge[bend left=50] (incr1)
            (incr1) edge (null3)
            (for1) edge[bend left=60] (blob2)
            (blob2) edge (null4)
            (for1) edge[bend left=80] (for2)

            (for2) edge (init2)
            (init2) edge (null5)
            (for2) edge (term2)
            (term2) edge (null6)
            (for2) edge[bend left=50] (incr2)
            (incr2) edge (null7)
            (for2) edge[bend left=40] (blob3)
            (blob3) edge (for3)
            (for2) edge[bend left=50] (blob6)
            (blob6) edge (null13)

            (for3) edge (init3)
            (init3) edge (null8)
            (for3) edge (term3)
            (term3) edge (null9)
            (for3) edge[bend left=50] (incr3)
            (incr3) edge (null10)
            (for3) edge[bend left=40] (blob4)
            (blob4) edge (null11)
            (for3) edge[bend left=40] (blob5)
            (blob5) edge (null12);
        \end{tikzpicture}}
    \end{subfigure}
    \hfill
    \begin{subfigure}[t]{0.48\textwidth}
        \centering
        \raisebox{35mm}{\begin{tikzpicture}[node distance=2mm]
            \node[draw, align=left] (blob1){\texttt{int[] A = new int[m];}};
            \node[draw, align=left, below= of blob1] (init1) {\texttt{int i=0;}};
            \node[draw, align=left, below=2mm  of init1] (while1) {\texttt{while}};
            \node[draw, align=left, below left=2mm and -8mm of while1] (term1) {\texttt{i<A.length;}};
            \node[draw, align=left, below left=2mm and -10mm of term1] (null1) {\texttt{null}};
            \node[draw, align=left, below right=2mm and -20mm of null1] (blob2) {\texttt{A[i] = getInputFromUser();}};
            \node[draw, align=left, below= of blob2] (incr1) {\texttt{i++}};
            \node[draw, align=left, below left=2mm and -10mm of incr1] (null2) {\texttt{null}};
            \node[draw, align=left, below right=2mm and -8mm of null2] (init2) {\texttt{int j=0;}};
            \node[draw, align = left, below left  =2mm and -10mm of init2] (while2) {\texttt{while}};
            \node[draw, align=left, below left=2mm and -8mm of while2] (term2) {\texttt{j<A.length;}};
            \node[draw, align=left, below left=2mm and -10mm of term2] (null3) {\texttt{null}};
            \node[draw, align=left, below right=2mm and -8mm of null3] (blob3) {
                \texttt{int minIndex = j;}\\
                \texttt{int minValue = A[minIndex];}};
            \node[draw, align=left, below= of blob3] (init3) {\texttt{int k=j;}};
            \node[draw, align=left, below=2mm  of init3] (while3) {\texttt{while}};
            \node[draw, align=left, below left=2mm and -8mm of while3] (term3) {\texttt{k<A.length;}};
            \node[draw, align=left, below left=2mm and -10mm of term3] (null4) {\texttt{null}};
            \node[draw, align=left, below right=2mm and -8mm of null4] (blob4) {
                \texttt{if(A[k] < minValue) \{}\\
                \texttt{\quad minIndex = k;}\\
                \texttt{\quad minValue = A[k];}\\
                \texttt{\}}};
            \node[draw, align=left, below= of blob4] (incr3) {\texttt{k++}};
            \node[draw, align=left, below left=2mm and -10mm of incr3] (null5) {\texttt{null}};
            \node[draw, align=left, below right=2mm and -8mm of null5] (blob5) {
                \texttt{int temp = A[j];}\\
                \texttt{A[j] = A[k];}\\
                \texttt{A[k] = temp;}};
            \node[draw, align=left, below= of blob5] (incr2) {\texttt{j++}};
            \node[draw, align=left, below left=2mm and -10mm of incr2] (null6) {\texttt{null}};
            \node[draw, align=left, below right=2mm and -8mm of null6] (blob6) {\texttt{return A;}};
            \node[draw, align=left, below left=2mm and -10mm of blob6] (null7) {\texttt{null}};
            
            \draw[->]
            (blob1) edge (init1)
            (init1) edge (while1)

            (while1) edge (term1)
            (term1) edge (null1)
            (blob2) edge (incr1)
            (while1) edge[bend left = 60] (blob2)
            (incr1) edge (null2)
            (while1) edge[bend left = 90] (init2)
            (init2) edge (while2)

            (while2) edge (term2)
            (while2) edge[bend left = 40] (blob3)
            (term2) edge (null3)
            (blob3) edge (init3)
            (init3) edge (while3)

            (while3) edge (term3)
            (term3) edge (null4)
            (while3) edge[bend left = 40] (blob4)
            (blob4) edge (incr3)
            (incr3) edge (null5)
            (while3) edge[bend left = 60] (blob5)
            (blob5) edge (incr2)
            (incr2) edge (null6)
            (while2) edge[bend left = 70] (blob6)
            (blob6) edge (null7);
        \end{tikzpicture}}
    \end{subfigure}
    \caption{Source and target code snippets from Fig~\ref{fig:code-snippets} as trees.}\label{fig:source-code-snippet-as-tree}
\end{figure}
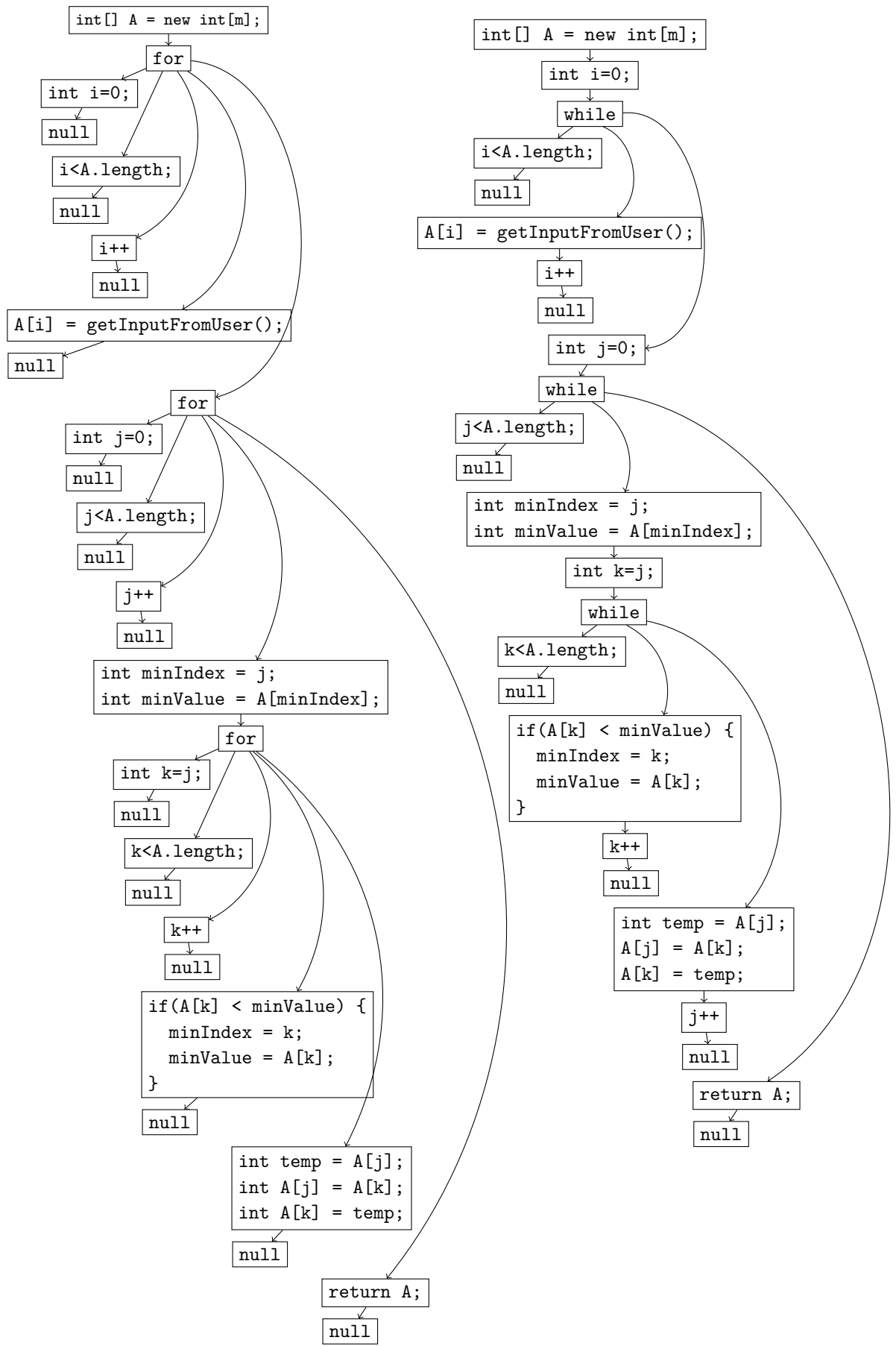
\clearpage

We would like to capture these codes with  patterns containing variables, which will allow us to move around the sections of code as needed. Let $\variables := \set{x^\examplePatternLanguageInit_i/1, x^\examplePatternLanguageCond_i/1, x^\examplePatternLanguageIncr_i/1 \mid i \in [n]}$ be the set of variables. Recall that $n$, bound on the number of $\examplePatternLanguageFor$ loops in our programs, is a constant and hence we have only a fixed finite number of variables.

\begin{figure}[!ht]
    \centering
    \scalebox{0.8}{\begin{tikzpicture}[node distance = 2mm]
        \node[draw, align = left] (blob1) {\texttt{int[] A = new int[m];}}; 
        \node[draw, color=blue, align=left, below =2mm of blob1] (loop1) {$\exampleSeedLanguageLoop_1$};
        \node[draw, align=left, below left=2mm and -10mm of loop1] (blob2) {\texttt{A[i] = getInputFromUser();}};
        \node[draw, align=left, color=red, below=of blob2] (endLoop1) {$\exampleSeedLanguageEndLoop_1$};
        \node[draw, align=left, color=blue, below right=2mm and 8mm of endLoop1] (loop2) {$\exampleSeedLanguageLoop_2$};
        \node[draw, align=left, below left=2mm and -10mm of loop2] (blob3) {
            \texttt{int minIndex = j;}\\
            \texttt{int minValue = A[minIndex];}};
        \node[draw, color=blue, align=left, below =2mm of blob3] (loop3) {$\exampleSeedLanguageLoop_3$};
        \node[draw, align=left, below left=2mm and -10mm of loop3] (blob4) {
            \texttt{if(A[k] < minValue) \{}\\
            \texttt{minIndex = k;}\\
            \texttt{minValue = A[k]; \}}};
        \node[draw, align=left, color=red, below=of blob4] (endLoop3) {$\exampleSeedLanguageEndLoop_3$};
        \node[draw, align=left, below right =2mm and 2mm of endLoop3] (blob5) {
            \texttt{int temp = A[j];}\\
            \texttt{A[j] = A[k];}\\
            \texttt{A[k] = temp;}};
        \node[draw, align=left, color=red, below=of blob5] (endLoop2) {$\exampleSeedLanguageEndLoop_2$};
        \node[draw, align=left, below right=2mm and -2mm of endLoop2] (blob6) {\texttt{return A;}};
        \node[draw, align=left, below=2mm of blob6] (null1) {$\examplePatternLanguageNull$};
        
        \draw[->]
        (blob1) edge (loop1)
        (loop1) edge[bend left=0] (blob2)
        (blob2) edge (endLoop1)
        (loop1) edge[bend left=50] (loop2)
        (loop2) edge (blob3)
        (blob3) edge (loop3)
        (loop3) edge (blob4)
        (blob4) edge (endLoop3)
        (loop3) edge[bend left=50] (blob5)
        (blob5) edge (endLoop2)
        (loop2) edge[bend left=50] (blob6)
        (blob6) edge (null1);
    \end{tikzpicture}}    
    \caption{The seed tree generating the source and target patterns in Figure~\ref{fig:sourceTargetPatternTrees}.}\label{fig:seedTreeGeneratingForWhilePatterns}
\end{figure}
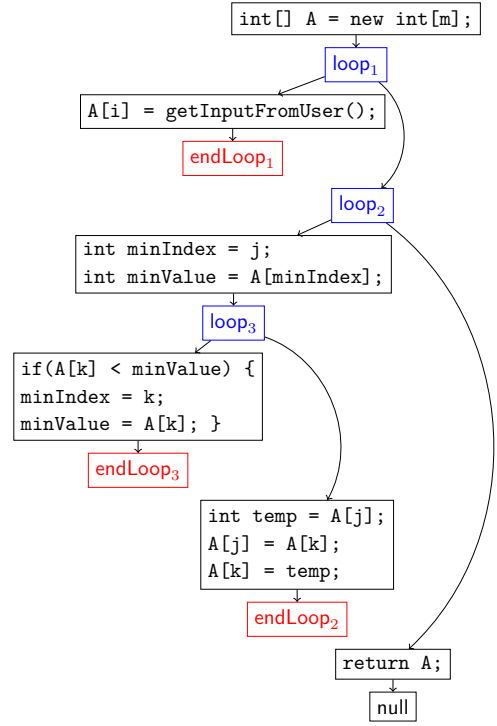

\begin{figure}[!ht]
    \scalebox{0.73}{\begin{tikzpicture}[node distance=2mm]
        \node[draw, align=left] (blob1) {\texttt{int[] A = new int[m];}};
        \node[draw, align=left, color=blue, below=2mm  of blob1] (for1) {\texttt{for}};
        \node[draw, align=left, below left= of for1, color=blue] (init1) {\texttt{$x^\examplePatternLanguageInit_1$}};
        \node[draw, align=left, color=blue, below left=2mm and -10mm of init1] (null1) {\texttt{$\examplePatternLanguageNull$}};
        \node[draw, align=left, color=blue, below right=2mm and -8mm of null1] (term1) {\texttt{$x^\examplePatternLanguageCond_1$}};
        \node[draw, align=left, color=blue, below left=2mm and -10mm of term1] (null2) {\texttt{$\examplePatternLanguageNull$}};
        \node[draw, align=left, color=blue, below right=2mm and -3mm of null2] (incr1) {\texttt{$x^\examplePatternLanguageIncr_1$}};
        \node[draw, align=left, color = blue, below left  =2mm and -10mm of incr1] (null3) {\texttt{$\examplePatternLanguageNull$}};
        \node[draw, align=left, below right = 2mm and -28mm of null3] (blob2){\texttt{A[i] = getInputFromUser();}};
        \node[draw, align=left, color=red, below right=2mm and -20mm of blob2] (null4) {\texttt{$\examplePatternLanguageNull$}};

        \node[draw, align=left, color=blue, below left=2mm and -5mm of null4] (for2){\texttt{for}};
        \node[draw, align=left, below left= of for2, color=blue] (init2) {\texttt{$x^\examplePatternLanguageInit_2$}};
        \node[draw, align=left, color=blue, below left=2mm and -10mm of init2] (null5) {\texttt{$\examplePatternLanguageNull$}};
        \node[draw, align=left, color=blue, below right=2mm and -8mm of null5] (term2) {\texttt{$x^\examplePatternLanguageCond_2$}};
        \node[draw, align=left, color=blue, below left=2mm and -10mm of term2] (null6) {\texttt{$\examplePatternLanguageNull$}};
        \node[draw, align=left, color=blue, below right=2mm and -3mm of null6] (incr2) {\texttt{$x^\examplePatternLanguageIncr_2$}};
        \node[draw, align=left, color = blue, below=2mm and -10mm of incr2] (null7) {\texttt{$\examplePatternLanguageNull$}};
        \node[draw, align=left, below right = 2mm and -28mm of null7] (blob3){
            \texttt{int minIndex = j;}\\
            \texttt{int minValue = A[minIndex];}};

        \node[draw, align=left, color=blue, below=2mm  of blob3] (for3) {\texttt{for}};
        \node[draw, align=left, below left= of for3, color=blue] (init3) {\texttt{$x^\examplePatternLanguageInit_3$}};
        \node[draw, align=left, color=blue, below left=2mm and -10mm of init3] (null8) {\texttt{$\examplePatternLanguageNull$}};
        \node[draw, align=left, color=blue, below right=2mm and -8mm of null8] (term3) {\texttt{$x^\examplePatternLanguageCond_3$}};
        \node[draw, align=left, color=blue, below left=2mm and -10mm of term3] (null9) {\texttt{$\examplePatternLanguageNull$}};
        \node[draw, align=left, color=blue, below right=2mm and -3mm of null9] (incr3) {\texttt{$x^\examplePatternLanguageIncr_3$}};
        \node[draw, align=left, color = blue, below left  =2mm and -10mm of incr3] (null10) {\texttt{$\examplePatternLanguageNull$}};
        \node[draw, align=left, below right = 2mm and -28mm of null10] (blob4){
            \texttt{if(A[k] < minValue) \{}\\
            \texttt{minIndex = k;}\\
            \texttt{minValue = A[k]; \}}};
        \node[draw, align=left, color=red, below right=2mm and -16mm of blob4] (null11) {\texttt{$\examplePatternLanguageNull$}};
        \node[draw, align=left, below left=2mm and -26mm of null11] (blob5){
            \texttt{int temp = A[j];}\\
            \texttt{A[j] = A[k];}\\
            \texttt{A[k] = temp;}};
        \node[draw, align=left, color=red, below left=2mm and -10mm of blob5] (null12) {\texttt{$\examplePatternLanguageNull$}};
        \node[draw, align=left, below right=2mm and 10mm of null12] (blob6){\texttt{return A;}};
        \node[draw, align=left, below left=2mm and -10mm of blob6] (null13) {\texttt{$\examplePatternLanguageNull$}};

        \draw[->]
        (blob1) edge (for1)
        (for1) edge[color=blue] (init1)
        (init1) edge[color=blue] (null1)
        (for1) edge[color=blue, bend left=30] (term1)
        (term1) edge[color=blue] (null2)
        (for1) edge[bend left=30, color=blue] (incr1)
        (incr1) edge[color=blue] (null3)
        (for1) edge[bend left=50] (blob2)
        (blob2) edge (null4)
        (for1) edge[bend left=60] (for2)

        (for2) edge[color=blue] (init2)
        (init2) edge[color=blue] (null5)
        (for2) edge[color=blue, bend left=30] (term2)
        (term2) edge[color=blue] (null6)
        (for2) edge[bend left=30, color=blue] (incr2)
        (incr2) edge[color=blue] (null7)
        (for2) edge[bend left=50] (blob3)
        (blob3) edge (for3)

        (for3) edge[color=blue] (init3)
        (init3) edge[color=blue] (null8)
        (for3) edge[color=blue, bend left=30] (term3)
        (term3) edge[color=blue] (null9)
        (for3) edge[bend left=30, color=blue] (incr3)
        (incr3) edge[color=blue] (null10)
        (for3) edge[bend left=50] (blob4)
        (blob4) edge (null11)
        (for3) edge[bend left=40] (blob5)
        (blob5) edge (null12)
        
        (for2) edge[bend left=30] (blob6)
        (blob6) edge (null13);

        \pgfmathsetmacro{\sourceLoopWidth}{5.5}
        \pgfmathsetmacro{\sourceLoopLeft}{-2.4}
        \pgfmathsetmacro{\sourceTextOffset}{0.25}

        \pgfmathsetmacro{\sourceRedLoopHeight}{0.7}
        \pgfmathsetmacro{\sourceRedTextOffsetLeft}{1.3}
        \pgfmathsetmacro{\sourceRedLoopOneBottom}{-6.95}
        \pgfmathsetmacro{\sourceRedLoopTwoBottom}{-23.2}
        \pgfmathsetmacro{\sourceRedLoopThreeBottom}{-20.85}

        \pgfmathsetmacro{\sourceBlueLoopHeight}{5.15}
        \pgfmathsetmacro{\sourceBlueTextOffsetLeft}{1}
        \pgfmathsetmacro{\sourceBlueLoopOneBottom}{-5.5}
        \pgfmathsetmacro{\sourceBlueLoopTwoBottom}{-12.15}
        \pgfmathsetmacro{\sourceBlueLoopThreeBottom}{-18.5}

        \begin{scope}[on background layer={}]
            \draw[rounded corners=3mm, color=blue, dashed, thick] (\sourceLoopLeft, \sourceBlueLoopOneBottom) rectangle (\sourceLoopLeft + \sourceLoopWidth, \sourceBlueLoopOneBottom + \sourceBlueLoopHeight);
            \draw[color=blue] (\sourceLoopLeft + \sourceBlueTextOffsetLeft, \sourceBlueLoopOneBottom + \sourceTextOffset) node {$\phi_\src(\exampleSeedLanguageLoop_1)$};
            \draw[rounded corners=3mm, color=blue, dashed, thick] (\sourceLoopLeft, \sourceBlueLoopTwoBottom) rectangle (\sourceLoopLeft + \sourceLoopWidth, \sourceBlueLoopTwoBottom + \sourceBlueLoopHeight);
            \draw[color=blue] (\sourceLoopLeft + \sourceBlueTextOffsetLeft, \sourceBlueLoopTwoBottom + \sourceTextOffset) node {$\phi_\src(\exampleSeedLanguageLoop_2)$};
            \draw[rounded corners=3mm, color=blue, dashed, thick] (\sourceLoopLeft, \sourceBlueLoopThreeBottom) rectangle (\sourceLoopLeft + \sourceLoopWidth, \sourceBlueLoopThreeBottom + \sourceBlueLoopHeight);
            \draw[color=blue] (\sourceLoopLeft + \sourceBlueTextOffsetLeft, \sourceBlueLoopThreeBottom + \sourceTextOffset) node {$\phi_\src(\exampleSeedLanguageLoop_3)$};

            \draw[rounded corners=3mm, color=red, dashed, thick] (\sourceLoopLeft, \sourceRedLoopOneBottom) rectangle (\sourceLoopLeft + \sourceLoopWidth, \sourceRedLoopOneBottom + \sourceRedLoopHeight);
            \draw[color=red] (\sourceLoopLeft + \sourceRedTextOffsetLeft, \sourceRedLoopOneBottom + \sourceTextOffset) node {$\phi_\src(\exampleSeedLanguageEndLoop_1)$};
            \draw[rounded corners=3mm, color=red, dashed, thick] (\sourceLoopLeft, \sourceRedLoopThreeBottom) rectangle (\sourceLoopLeft + \sourceLoopWidth, \sourceRedLoopThreeBottom + \sourceRedLoopHeight);
            \draw[color=red] (\sourceLoopLeft + \sourceRedTextOffsetLeft, \sourceRedLoopThreeBottom + \sourceTextOffset) node {$\phi_\src(\exampleSeedLanguageEndLoop_3)$};
            \draw[rounded corners=3mm, color=red, dashed, thick] (\sourceLoopLeft, \sourceRedLoopTwoBottom) rectangle (\sourceLoopLeft + \sourceLoopWidth, \sourceRedLoopTwoBottom + \sourceRedLoopHeight);
            \draw[color=red] (\sourceLoopLeft + \sourceRedTextOffsetLeft, \sourceRedLoopTwoBottom + \sourceTextOffset) node {$\phi_\src(\exampleSeedLanguageEndLoop_2)$};
        \end{scope}
    \end{tikzpicture}}
    \hfill
    \scalebox{0.73}{\begin{tikzpicture}[node distance=2mm]
        \node[draw, align=left] (blob1) {\texttt{int[] A = new int[m];}}; 
        \node[draw, align=left, color=blue, below =2mm of blob1] (init1) {$x^\examplePatternLanguageInit_1$};
        \node[draw, align=left, color=blue, below=2mm of init1] (while1) {\texttt{while}};
        \node[draw, align=left, color=blue, below left=2mm and -8mm of while1] (term1) {$x^\examplePatternLanguageCond_1$};	
        \node[draw, align=left, color=blue, below left=2mm and -10mm of term1] (null1) {$\examplePatternLanguageNull$};
        \node[draw, align=left, below right=2mm and -20mm of null1] (blob2) {\texttt{A[i] = getInputFromUser();}};
        \node[draw, align=left, color=red, below left=2mm and -22mm of blob2] (incr1) {$x^\examplePatternLanguageIncr_1$};
        \node[draw, align=left, color=red, below left=2mm and -10mm of incr1] (null2) {$\examplePatternLanguageNull$};
        \node[draw, align=left, color=blue, below right=2mm and -8mm of null2] (init2) {$x^\examplePatternLanguageInit_2$};

        \node [draw, align = left, color=blue, below left  =2mm and -10mm of init2] (while2) {\texttt{while}};
        \node[draw, align=left, color=blue, below left=2mm and -8mm of while2] (term2) {$x^\examplePatternLanguageCond_2$};	
        \node[draw, align=left, color=blue, below left=2mm and -10mm of term2] (null3) {$\examplePatternLanguageNull$};
        \node[draw, align=left, below right=2mm and -30mm of null3] (blob3) {
            \texttt{int minIndex = j;}\\
            \texttt{int minValue = A[minIndex];}};
        \node[draw, align=left, color=blue, below =2mm of blob3] (init3) {$x^\examplePatternLanguageInit_3$};
        \node[draw, align=left, color=blue, below=2mm of init3] (while3) {\texttt{while}};
        \node[draw, align=left, color=blue, below left=2mm and -8mm of while3] (term3) {$x^\examplePatternLanguageCond_3$};	
        \node[draw, align=left, color=blue, below left=2mm and -10mm of term3] (null4) {$\examplePatternLanguageNull$};
        \node[draw, align=left, below right=2mm and -20mm of null4] (blob4) {
            \texttt{if(A[k] < minValue) \{}\\
            \texttt{minIndex = k;}\\
            \texttt{minValue = A[k]; \}}};
        \node[draw, align=left, color=red, below left=2mm and -22mm of blob4] (incr3) {$x^\examplePatternLanguageIncr_3$};
        \node[draw, align=left, color=red, below left=2mm and -10mm of incr3] (null5) {$\examplePatternLanguageNull$};
        \node[draw, align=left, below right=2mm and -15mm of null5] (blob5) {
            \texttt{int temp = A[j];}\\
            \texttt{A[j] = A[k];}\\
            \texttt{A[k] = temp;}};
        \node[draw, align=left, color=red, below left=2mm and -22mm of blob5] (incr2) {$x^\examplePatternLanguageIncr_2$};
        \node[draw, align=left, color=red, below left=2mm and -10mm of incr2] (null6) {$\examplePatternLanguageNull$};
        \node[draw, align=left, below right=2mm and -15mm of null6] (blob6) {\texttt{return A;}};
        \node[draw, align=left, below left=2mm and -10mm of blob6] (null7) {$\examplePatternLanguageNull$};
        
        \draw[->]
        (blob1) edge (init1)
        
        (init1) edge[color=blue] (while1)
        (while1) edge[color=blue] (term1)
        (term1) edge[color=blue] (null1)
        (blob2) edge (incr1)
        (while1) edge[bend left=40] (blob2)
        (incr1) edge[color=red] (null2)
        (while1) edge[bend left=90] (init2)
        
        (init2) edge[color=blue] (while2)
        (while2) edge[color=blue] (term2)
        (term2) edge[color=blue] (null3)
        (while2) edge[bend left=50] (blob3)
        (blob3) edge (init3)

        (init3) edge[color=blue] (while3)
        (while3) edge[color=blue] (term3)
        (term3) edge[color=blue] (null4)
        (while3) edge[bend left=40] (blob4)
        (blob4) edge (incr3)
        (incr3) edge[color=red] (null5)
        (while3) edge[bend left=70] (blob5)
        (blob5) edge (incr2)
        (incr2) edge[color=red] (null6)
        (while2) edge[bend left=70] (blob6)
        (blob6) edge (null7);

        \pgfmathsetmacro{\targetLoopWidth}{4.8}
        \pgfmathsetmacro{\targetLoopLeft}{-2.8}
        \pgfmathsetmacro{\targetTextOffset}{0.25}

        \pgfmathsetmacro{\targetRedLoopHeight}{1.45}
        \pgfmathsetmacro{\targetRedTextOffsetLeft}{1.3}
        \pgfmathsetmacro{\targetRedLoopOneBottom}{-5.6}
        \pgfmathsetmacro{\targetRedLoopTwoBottom}{-19.2}
        \pgfmathsetmacro{\targetRedLoopThreeBottom}{-16}

        \pgfmathsetmacro{\targetBlueLoopHeight}{3}
        \pgfmathsetmacro{\targetBlueTextOffsetLeft}{1}
        \pgfmathsetmacro{\targetBlueLoopOneBottom}{-3.35}
        \pgfmathsetmacro{\targetBlueLoopTwoBottom}{-8.65}
        \pgfmathsetmacro{\targetBlueLoopThreeBottom}{-12.85}

        \begin{scope}[on background layer={}]
            \draw[rounded corners=3mm, color=blue, dashed, thick] (\targetLoopLeft, \targetBlueLoopOneBottom) rectangle (\targetLoopLeft + \targetLoopWidth, \targetBlueLoopOneBottom + \targetBlueLoopHeight);
            \draw[color=blue] (\targetLoopLeft + \targetBlueTextOffsetLeft, \targetBlueLoopOneBottom + \targetTextOffset) node {$\phi_\tgt(\exampleSeedLanguageLoop_1)$};
            \draw[rounded corners=3mm, color=blue, dashed, thick] (\targetLoopLeft, \targetBlueLoopTwoBottom) rectangle (\targetLoopLeft + \targetLoopWidth, \targetBlueLoopTwoBottom + \targetBlueLoopHeight);
            \draw[color=blue] (\targetLoopLeft + \targetBlueTextOffsetLeft, \targetBlueLoopTwoBottom + \targetTextOffset) node {$\phi_\tgt(\exampleSeedLanguageLoop_2)$};
            \draw[rounded corners=3mm, color=blue, dashed, thick] (\targetLoopLeft, \targetBlueLoopThreeBottom) rectangle (\targetLoopLeft + \targetLoopWidth, \targetBlueLoopThreeBottom + \targetBlueLoopHeight);
            \draw[color=blue] (\targetLoopLeft + \targetBlueTextOffsetLeft, \targetBlueLoopThreeBottom + \targetTextOffset) node {$\phi_\tgt(\exampleSeedLanguageLoop_3)$};

            \draw[rounded corners=3mm, color=red, dashed, thick] (\targetLoopLeft, \targetRedLoopOneBottom) rectangle (\targetLoopLeft + \targetLoopWidth, \targetRedLoopOneBottom + \targetRedLoopHeight);
            \draw[color=red] (\targetLoopLeft + \targetRedTextOffsetLeft, \targetRedLoopOneBottom + \targetTextOffset) node {$\phi_\tgt(\exampleSeedLanguageEndLoop_1)$};
            \draw[rounded corners=3mm, color=red, dashed, thick] (\targetLoopLeft, \targetRedLoopThreeBottom) rectangle (\targetLoopLeft + \targetLoopWidth, \targetRedLoopThreeBottom + \targetRedLoopHeight);
            \draw[color=red] (\targetLoopLeft + \targetRedTextOffsetLeft, \targetRedLoopThreeBottom + \targetTextOffset) node {$\phi_\tgt(\exampleSeedLanguageEndLoop_3)$};
            \draw[rounded corners=3mm, color=red, dashed, thick] (\targetLoopLeft, \targetRedLoopTwoBottom) rectangle (\targetLoopLeft + \targetLoopWidth, \targetRedLoopTwoBottom + \targetRedLoopHeight);
            \draw[color=red] (\targetLoopLeft + \targetRedTextOffsetLeft, \targetRedLoopTwoBottom + \targetTextOffset) node {$\phi_\tgt(\exampleSeedLanguageEndLoop_2)$};
        \end{scope}
    \end{tikzpicture}}
    \caption{Source and target patterns. The source and target morphisms are identity everywhere, except for $\exampleSeedLanguageLoop_i$ and $\exampleSeedLanguageEndLoop_i$,  which  are indicated by colors blue and red respectively.}\label{fig:sourceTargetPatternTrees}
\end{figure}

The source and target patterns are given in Figure~\ref{fig:sourceTargetPatternTrees} and a seed tree generating these patterns is given in Figure~\ref{fig:seedTreeGeneratingForWhilePatterns}.
Note that the seed alphabet is given by $\seedAlphabet := \coreAlphabet \uplus \set{\exampleSeedLanguageLoop_i/2, \exampleSeedLanguageEndLoop_i/0 \mid i \in [n]} \uplus \set{\exampleSeedLanguageNull/0}$. 

In general, the seed language will contain all  trees in $\trees{\seedAlphabet}$ which contain each of the letters $\exampleSeedLanguageLoop_i$ and $\exampleSeedLanguageEndLoop_i$ at most once and are well-bracketed. Let us take a detailed look at how the language looks like. Note that we a priori limit the scope of code we transform to instances having at most $n$ loops, where $n$ is a fixed natural number. We start with the various alphabets we use:
\begin{itemize}
    \item Tree alphabet $\alphabet := \coreAlphabet \uplus \set{\examplePatternLanguageFor/5, \examplePatternLanguageWhile/3, \examplePatternLanguageNull/0}$, where $\coreAlphabet$ is, for example, all the symbols available on a standard keyboard, each with an arity of 1.
    \item Pattern alphabet $\alphabet \uplus \variables$, where $\variables := \set{x^\examplePatternLanguageInit_i/1, x^\examplePatternLanguageCond_i/1, x^\examplePatternLanguageIncr_i/1 \mid i \in [n]}$.
    \item Seed alphabet $\seedAlphabet := \coreAlphabet \uplus \set{\exampleSeedLanguageLoop_i/2, \exampleSeedLanguageEndLoop_i/0 \mid i \in [n]} \uplus \set{\examplePatternLanguageNull/0}$. 
\end{itemize}
In the seed alphabet, we have binary letters $\exampleSeedLanguageLoop_i$ for all $i \in [n]$. We interpret this as follows: the letter $\exampleSeedLanguageLoop_i$ is a placeholder for the $i^{th}$ loop in the program. The left child of the letter will be a tree representing the part of the code that is \emph{inside} the loop, ending with an $\exampleSeedLanguageEndLoop_i$, which is a leaf letter to signify the end of the loop. The right child will be a tree representing the part of the code that follows \emph{after} the loop, eventually ending with a $\examplePatternLanguageNull$ leaf letter. Hence in the loop language, we would want that each pair of letters $(\exampleSeedLanguageLoop_i, \exampleSeedLanguageEndLoop_i)$ which signify loop initiation and loop termination respectively to appear at most once. This is a necessary `hygiene' condition upon which we will build more structure.

For this condition to hold, we define for each $i \in [n]$, a finite-state tree automaton $\uniqueLetterAutomaton{i} := (\set{q_0, q_1, q_2, q_3, q_4}, \seedAlphabet, \Delta_i, \set{q_0, q_3})$ where $\Delta_i$ is given by
\begin{align*}
    \Delta_i :=
    & \set{(\exampleSeedLanguageEndLoop_i, q_1), (\examplePatternLanguageNull, q_0)} \cup \set{(\exampleSeedLanguageEndLoop_j, q_0) \mid j \in [n], j \neq i} \\
    & \cup \set{(q_x, q_y, \exampleSeedLanguageLoop_i, q_{\min(4, x+y+2)})} \cup \set{(q_x, a, q_x) \mid x \in \set{0, 1, 2, 3, 4}} \\
    & \cup \set{(q_x, q_y, \exampleSeedLanguageLoop_j, q_{\min(4, x+y)}) \mid j \in [n], j \neq i, (x, y) \neq (1, 1)} \\
    & \cup \set{(q_1, q_1, \exampleSeedLanguageLoop_j, q_4)}
\end{align*}
We see that each of the automata $\uniqueLetterAutomaton{i}$ accepts all trees $s$ over $\seedAlphabet$ such that the letters $\exampleSeedLanguageLoop_i, \exampleSeedLanguageEndLoop_i$ both appear either once each, or not at all in $s$. Thus having $S \subseteq \left[ \bigcap_{i=1}^n L(\uniqueLetterAutomaton{i}) \right]$ implies that every tree $s$ contains the pair of letters $\exampleSeedLanguageLoop_i, \exampleSeedLanguageEndLoop_i$ at most once.

On top of this requirement, we also need that every tree in the seed language be ``well-nested'' in terms of the containment of loops. For this, we define an automaton $\automaton := (Q, \seedAlphabet, \Delta, F)$ as:
\begin{align*}
    Q := & \set{q_i \mid i \in [n]} \cup \set{q_0}\\
    F := & \set{q_0}\\
    \Delta
    := & \set{(\examplePatternLanguageNull, q_0)} \cup \set{(\exampleSeedLanguageEndLoop_i, q_i) \mid i \in [n]} \cup\\
    & \set{(q_i, q_j, \exampleSeedLanguageLoop_i, q_j) \mid i, j \in [n] \cup \set{0}} \set{(q, a, q) \mid q \in Q, a \in \coreAlphabet}
\end{align*}
We demonstrate the working of this automaton on an example word we would like to have in the seed language. Below, we consider a tree $s$ over the seed alphabet $\seedAlphabet$ for the case $n=9$, along with a run of the automaton $\automaton$ shown in parenthesis, in blue. Note that the dashed edges represent the presence of any number of unary letters from $\coreAlphabet$; since they ``bubble up'' any state, they can be represented as such for purposes of better understanding.
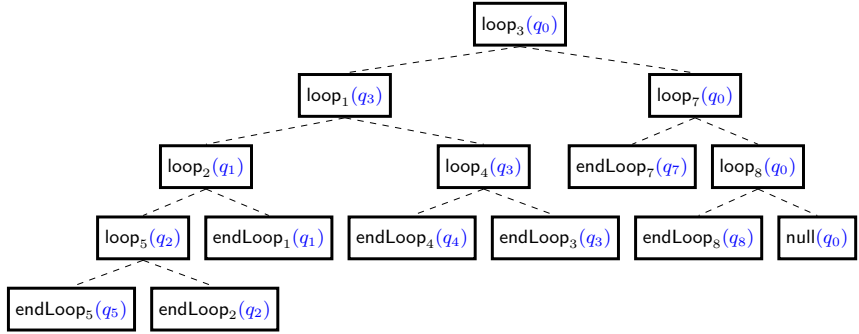
\[
\tikzset{
    Decision/.style = {draw, line width=1.4pt},
    Lottery/.style = {draw, line width=1.4pt},
    Outcome/.style = {circle, minimum width=3pt, fill, inner sep=0pt}
}
\scalebox{0.8}{
\begin{forest}
    for tree={grow=270, Decision}
    [$\exampleSeedLanguageLoop_3\blue{(q_0)}$
        [$\exampleSeedLanguageLoop_1\blue{(q_3)}$, edge=dashed
            [$\exampleSeedLanguageLoop_2\blue{(q_1)}$, edge=dashed
                [$\exampleSeedLanguageLoop_5\blue{(q_2)}$, edge=dashed
                    [$\exampleSeedLanguageEndLoop_5\blue{(q_5)}$, edge=dashed]
                    [$\exampleSeedLanguageEndLoop_2\blue{(q_2)}$, edge=dashed]
                ]
                [$\exampleSeedLanguageEndLoop_1\blue{(q_1)}$, edge=dashed]
            ]
            [$\exampleSeedLanguageLoop_4\blue{(q_3)}$, edge=dashed
                [$\exampleSeedLanguageEndLoop_4\blue{(q_4)}$, edge=dashed]
                [$\exampleSeedLanguageEndLoop_3\blue{(q_3)}$, edge=dashed]
            ]
        ]
        [$\exampleSeedLanguageLoop_7\blue{(q_0)}$, edge=dashed
            [$\exampleSeedLanguageEndLoop_7\blue{(q_7)}$, edge=dashed]
            [$\exampleSeedLanguageLoop_8\blue{(q_0)}$, edge=dashed
                [$\exampleSeedLanguageEndLoop_8\blue{(q_8)}$, edge=dashed]
                [$\examplePatternLanguageNull\blue{(q_0)}$, edge=dashed]
            ]
        ]
    ]
\end{forest}}
\]
Finally, we define the seed language to be $S := L(\automaton) \cap \left[ \bigcap_{i=1}^n L(\uniqueLetterAutomaton{i}) \right]$, which is a regular language.

In general, we see that for any program with no more than $n$ $\examplePatternLanguageFor$ loops, we will be able to find a seed tree, which will give us a source pattern and a target pattern through homomorphisms $\sourceMap, \targetMap$ respectively. The source pattern will match the program code with not more than $n$ $\examplePatternLanguageFor$ loops with a substitution, which when applied to the target pattern will give us the transformed version of the program code which has all $\examplePatternLanguageFor$ loops transformed into $\examplePatternLanguageWhile$ loops.

We may want to verify that applying a transformation to source trees belonging to some regular language (such as syntactically correct programs that use for loops) will never result in target trees violating desired properties (such as syntactically correct programs that use while loops). This motivates the problem we consider for our model, defined next.


\section{Decision Problems for Tree Transformations}\label{sec:problemsAndResults}

A simple question is to check whether there is a matching.
\begin{problemDefinition}
    \problemtitle{\matchingProblemFull}
    \probleminput{Pattern $\pattern$, tree $t$, guard $\guardFunction$. \grey{\slash\slash\ The guard languages are represented by NBTA.}}
    \problemquestion{Does $t$ match $\pattern$ under $\guardFunction$?}
\end{problemDefinition}

We will show that this problem is already \nptc.

\begin{theorem}[\cite{mtchLBound}]\label{thm:matching}
    The \matchingProblemFull\ is \nptc.
\end{theorem}

\begin{proof}
    \textbf{(Upper Bound)} Let $(\pattern, t, \guardFunction)$ be an instance of \matchingProblemFull, where $\pattern$ is a pattern over $(\alphabet, \variables)$, $t \in \trees{\alphabet}$ and $\guardFunction : \variables \to 2^{\trees{\alphabet}}$ where each value of the function is given by an NBTA. Let the size of the input be $N := \abs{\pattern} + \abs{t} + \abs{\guardFunction}$. Given a certificate, which is a particular substitution whose size is bounded by $|t|$, we can verify whether $\substitution(\pattern) = t$ in $\bigo(|t|)$ time, which is linear in the size of the input $N$. We can also verify that for every $x \in \variables(\pattern)$, $\substitution(x) \in \guardOf{x}$. Since there exists a polynomial-time verifier algorithm for the problem, we see that \matchingProblemFull\ is in \np.

    \textbf{(Lower Bound)} For this proof, we present a version adapted from~\cite{mtchLBound}. The proof proceeds by establishing a reduction from \oneInThreeSatFull, which is defined below.
    \begin{problemDefinition}
        \problemtitle{\oneInThreeSatFull}
        \probleminput{3CNF formula $\phi$ without negated variables.}
        \problemquestion{Is there a truth assignment to the variables of $\phi$ such that every clause contains exactly one literal that is true and $\phi$ evaluates to true?}
    \end{problemDefinition}

    We know from~\cite{DBLP:conf/stoc/OneInThreeSAT} that \oneInThreeSatFull\ is \nptc. Consider the following instance $\{J(v_{f(3i-2)}, v_{f(3i-1)}, v_{f(3i)}) \mid i \in [n]\}$ of \oneInThreeSatFull, where $V := \{v_1, \dots, v_k\}$ is a finite set of variables and $f : [3n] \to [k]$. We also have $J$ to be a function which evaluates to true iff exactly one of its arguments is set to true, and false otherwise.

    Let $\alphabet := \{a/1, \#/1, \lambda/0\}$ be a ranked alphabet and $\variables := \{x_i/1 \mid i \in \N\}$ be a set of ranked variables. Consider a pattern $\pattern$ over $(\alphabet, \variables)$ given by \[ x_{f(1)}x_{f(2)}x_{f(3)}\#\dots\#x_{f(3i-2)}x_{f(3i-1)}x_{f(3i)}\#\dots\#x_{f(3n-2)}x_{f(3n-1)}x_{f(3n)}\lambda. \] and a tree over $\alphabet$ given by $t := a\#a\#\dots\#a\#\dots\#a$, where the number of $a$'s in the tree is $n$. We see that the only substitution $\substitution : V \to \{a, \e\}$ which matches $t$ to $\alpha$ is the one which assigns each true variable to $a$ and every false variable to $\e$. The guard languages can be specified easily by a simple NBTA. Hence we get that a solution for the original instance of \oneInThreeSatFull\ exists if and only if a solution for the corresponding \matchingProblemFull\ has a solution. This completes the reduction.
\end{proof}

Next we consider the reachability problem for atomic transformations and establish it is \exptc.

\begin{problemDefinition}
    \problemtitle{\atomicTransformReachabilityFull}
    \probleminput{Atomic transformation $\atomicTransform := (\sourcePattern, \targetPattern, \guardFunction)$, source language $\sourceLanguage$,  target language $\targetLanguage$. \grey{\slash\slash\ The  languages are represented by NBTA.}}
    \problemquestion{Do we have that $\atomicTransform(\sourceLanguage) \cap \targetLanguage \neq \emptyset$?}
\end{problemDefinition}

\begin{theorem}\label{thm:typecheckingProblem}
    The \atomicTransformReachabilityFull\ is \exptc.
\end{theorem}

We give a proof which is a special case of the \metaTypecheckingProblemReformulationFull. Clearly, it is easy to see that this is indeed a case of \metaTypecheckingProblemReformulationFull\ when the seed language is singleton. If $\automaton$ is such that $L(\automaton) = \{\tau\}$, then we just define/rename patterns $\pattern_\src := \sourceMap(\tau), \pattern_\tgt := \targetMap(\tau)$.

\begin{proof}
    \textbf{(Upper Bound)} Let $\atomicTransform := (\sourcePattern, \targetPattern, \guardFunction)$ be an atomic action, and let $\sourceAutomatonB, \targetAutomatonB$ be NBTAs with $\sourceLanguage := L(\sourceAutomatonB)$ and $\targetLanguage := L(\targetAutomatonB)$. To determine whether $\atomicTransform[\sourceLanguage] \cap \targetLanguage \neq \emptyset$, we need to find a substitution $\substitution$ such that $\substitution(\sourcePattern) \in \sourceLanguage$ and $\substitution(\targetPattern) \in \targetLanguage$ under $\guardFunction$. We further simplify the problem at hand. Consider a bigger alphabet $\Gamma := \alphabet \uplus \{\#/2\}$. Let $\pattern$ be a pattern over $\Gamma$ and $\variables$ given by $\pattern := \#(\sourcePattern, \targetPattern)$. Let $\automatonB$ be an NBTA over $\Gamma$ which recognises all trees of the form $\#(p_{\mathit{src}}, p_{\mathit{tgt}})$ with $p_{\mathit{src}} \in \sourceLanguage, p_{\mathit{tgt}} \in \targetLanguage$. Then \atomicTransformReachabilityFull\ now reduces to finding a substitution $\substitution$ over $\Gamma$ conforming to $\guardFunction$ such that $\substitution(\pattern) \in \genericLanguage := L(\automatonB)$.

    Let $M := \max \{\arity(x_i) \mid x_i \in \variables(\pattern)\}$. Let us assume that each variable $x/k \in \variables(\pattern)$ appears $b(x)$ number of times in $\pattern$. By Theorem~\ref{thm:relationRealisabilityExptc}, for each variable $x$, we choose a pair $(R(x), c(x))$ such that $R(x) \subseteq Q^{k+1}$ is a relation realised by context $c(x) \in \contexts{\Gamma}{k}$ on $\automatonB$ under $\guardOf{x}$ with the added condition that $|R(x)| \leq b(x)$. Once we have all such ordered pairs for individual variables, we simply use the substitution $\substitution := c$ and check whether the resulting tree belongs to $T$. This can be done in time given by
    \[ \bigo(|\automatonB| \cdot \max_x|\guardOf{x}| + |\pattern| \cdot 2^{|\automatonB| \cdot \max_x|\guardOf{x}|}) \]
    The number of times we have to perform this check operation is bounded above by the value ${\left(\max_x\left| \binom{Q^{\deg(x)}}{b(x)} \right|\right)}^{|\variables(\pattern)|}$. Hence the total time required for this is given by
    \[ {\left(\max_x\left| \binom{Q^{\deg(x)}}{b(x)} \right|\right)}^{|\variables(\pattern)|} \cdot \bigo(|\automatonB| \cdot \max_x|\guardOf{x}| + |\pattern| \cdot 2^{|\automatonB| \cdot \max_x|\guardOf{x}|}) \]
    which we can see is exponential in the size of the input.

    \textbf{(Lower Bound)} We show a reduction from \intersectionNonemptinessFull\ to \atomicTransformReachabilityFull. Since we know that \intersectionNonemptinessFull\ is \expth, we will conclude that \atomicTransformReachabilityFull\ is also \expth. Let $\automaton_1, \dots, \automaton_n$ be an instance of the \intersectionNonemptinessFull, where each $\automaton_i := (Q_i, \alphabet, \Transitions_i, F_i)$ is DBTA over the alphabet $\alphabet$ for all $i \in [n]$. We construct a corresponding instance of \atomicTransformReachabilityFull\ as follows. Let $\Gamma := \alphabet \uplus \{\#/n, \flat/0\}$ be a larger alphabet. Consider the atomic action $\atomicTransform := (\#(\underbrace{x, x, \dots, x}_{n}), \flat, \guardFunction)$ with $\guardOf{x} := \trees{\alphabet}$. Let $\automatonB_{src}$ be a NBTA such that $L_{src} := L(\automatonB_{src}) := \{\#(t_1, t_2, \dots, t_n) \mid t_i \in \trees{\alphabet},~\forall i \in [n]\}$. Let $\automatonB_{tgt}$ be an NBTA such that $L_{tgt} := L(\automatonB_{tgt}) := \{\flat\}$.

    We can see that $\atomicTransform[L_{src}] \cap L_{tgt} \neq \emptyset$ if and only if there exists $t \in \trees{\alphabet}$ such that $t \in L(\automaton_i)$ for all $i \in [n]$. This completes the reduction, and thus the proof.
\end{proof}

Our central problem of study is whether  a tree in a regular target language can be reached by a given transformation, starting from a tree in the regular source language, which we show is decidable.

\begin{problemDefinition}
	\problemtitle{\metaTypecheckingProblemReformulationFull}
	\probleminput{Transformation  $\fullTransform := (\seedLanguage,
		\sourceMap, \targetMap, \guardFunction)$, source language $\sourceLanguage$,  target language $\targetLanguage$. \grey{\slash\slash\ The  languages are represented by NBTA.}}
	\problemquestion{Do we have that $\fullTransform(\sourceLanguage) \cap \targetLanguage \neq \emptyset$?}
\end{problemDefinition}

\begin{theorem}\label{thm:metaTypecheckingUpperBound}
	The \metaTypecheckingProblemReformulationFull\ is in 2-\expt. More precisely, it can be solved in time $2^{\mathsf{poly}(\mathsf{|input|})}$, where $\mathsf{poly}$ is a  polynomial of degree $\deg(\fullTransform)$. 
\end{theorem}

Since NBTA are closed under complementation (with an exponential blow-up, see~\cite{tata}) we get that the type-checking problem defined below is in 3-\expt. 

\begin{problemDefinition}
    \problemtitle{\metaTypecheckingProblemFull}
    \probleminput{Transformation  $\fullTransform := (\seedLanguage,
    \sourceMap, \targetMap, \guardFunction)$, source language $\sourceLanguage$,  target language $\targetLanguage$. \grey{\slash\slash\ The  languages are represented by NBTA.}}
    \problemquestion{Do we have that $\fullTransform(\sourceLanguage) \subseteq \targetLanguage$?}
\end{problemDefinition}

\begin{corollary}
	The \metaTypecheckingProblemFull\ is in 3-\expt. 
\end{corollary}	

Our decision procedure for the \metaTypecheckingProblemReformulationFull\ is automata-theoretic. We first give the main ideas behind our decision procedures in the next section.


\section{Core Idea Behind the Automata Construction}\label{sec:proofOutline}

\begin{figure}[!ht]
    \centering
    \scalebox{1.07}{
    \begin{tikzpicture}[node distance=5mm, inner sep=0pt, text width=2mm, font=\sffamily\small]
        \node[fill=blue, circle, label={[xshift=-3mm, yshift=-3mm]$\eta$}] (eta) {};
        \node[fill=green, below=of eta,  circle, label={[xshift=-3mm, yshift=-3mm]$\beta$}] (beta) {};
        \node[fill=magenta, left=of beta,  circle, label={[xshift=-3mm, yshift=-3mm]$\alpha$}] (alpha) {};
        \node[fill=orange, right=of beta,  circle, label={[xshift=-3mm, yshift=-3mm]$\gamma$}] (gamma) {};
        
        \draw[dotted, thick] ($(eta)+(0,+5mm)$)--(eta);
        \draw[dotted, thick] ($(alpha)+(-2mm,-5mm)$)--(alpha);
        \draw[dotted, thick] ($(alpha)+(+2mm,-5mm)$)--(alpha);
        \draw[dotted, thick] ($(beta)+(-3mm,-5mm)$)--(beta);
        \draw[dotted, thick] ($(beta)+(0,-5mm)$)--(beta);
        \draw[dotted, thick] ($(beta)+(+3mm,-5mm)$)--(beta);
        \draw[dotted, thick] ($(gamma)+(-2mm,-5mm)$)--(gamma);
        \draw[dotted, thick] ($(gamma)+(+2mm,-5mm)$)--(gamma);
        \draw[-]
        (eta) edge (alpha)
        (eta) edge (beta)
        (eta) edge (gamma);
        
        \node[above=1cm of eta, text width=1.5cm] () {Seed tree};
        \node[right=0.5cm of eta, text width=1.5cm] (arrow1) {$\xrightarrow{\genericMap := \sourceMap}$};
        
        \coordinate (alphap) at ($(eta) + (1.7cm,-1.3cm)$);
        \coordinate (betap) at  ($(eta) + (3cm,-1.3cm)$);
        \coordinate (gammap) at ($(eta) + (4.3cm,-1.3cm)$);
        
        \fill[color=blue!20] ($(eta) + (3cm,15mm)$) -- ($(eta) + (4.5cm,-1cm)$) -- ($(eta) + (1.5cm,-1cm)$);
        \draw[color=blue, thick] ($(eta) + (3cm,15mm)$) -- ($(eta) + (4.5cm,-1cm)$) -- ($(eta) + (1.5cm,-1cm)$) -- ($(eta) + (3cm,15mm)$);
        \fill[color=magenta!20] (alphap) -- ($(eta) + (2.2cm,-2cm)$) -- ($(eta) + (1.2cm,-2cm)$);
        \draw[color=magenta, thick] (alphap) -- ($(eta) + (2.2cm,-2cm)$) -- ($(eta) + (1.2cm,-2cm)$) -- (alphap);
        \fill[color=green!20] (betap)-- ($(eta) + (3.5cm,-2cm)$) -- ($(eta) + (2.5cm,-2cm)$);
        \draw[color=green, thick] (betap)-- ($(eta) + (3.5cm,-2cm)$) -- ($(eta) + (2.5cm,-2cm)$) -- (betap);
        \fill[color=orange!20] (gammap) -- ($(eta) + (4.8cm,-2cm)$) -- ($(eta) + (3.8cm,-2cm)$);
        \draw[color=orange, thick] (gammap) -- ($(eta) + (4.8cm,-2cm)$) -- ($(eta) + (3.8cm,-2cm)$) -- (gammap);

        \draw[dotted, thick] ($(eta)+(3cm,+18mm)$)-- ($(eta) + (3cm,15mm)$) ;
        \draw[dotted, thick] ($(eta)+(1.4cm,-23mm)$)--($(eta) + (1.5cm,-2cm)$);
        \draw[dotted, thick] ($(eta)+(+20mm,-23mm)$)--($(eta) + (1.9cm,-2cm)$);
        \draw[dotted, thick] ($(eta)+(+25mm,-23mm)$)--($(eta) + (2.8cm,-2cm)$);
        \draw[dotted, thick] ($(eta)+(+30mm,-23mm)$)--($(eta) + (3cm,-2cm)$);
        \draw[dotted, thick] ($(eta)+(+35mm,-23mm)$)--($(eta) + (3.2cm,-2cm)$);
        \draw[dotted, thick] ($(eta)+(4cm,-23mm)$)--($(eta) + (4.1cm,-2cm)$);
        \draw[dotted, thick] ($(eta)+(+46mm,-23mm)$)--($(eta) + (4.5cm,-2cm)$);

        \draw[] ($(alphap)+(0,3mm)$)--(alphap) ;
        \draw[] ($(betap)+(0,3mm)$)--(betap) ;
        \draw[] ($(gammap)+(0,3mm)$)--(gammap) ;

        \node[fill=red, above right=5mm and 28mm of eta,  circle, label={[xshift=-3mm, yshift=-3mm]$x$}] (x1) {};
        \node[fill=red, below left=8mm and 4mm of x1,  circle, label={[xshift=-3mm, yshift=-3mm]$x$}] (x2) {};

        \draw[dotted, thick] ($(x1)+(-2mm,-5mm)$)--(x1) ;
        \draw[dotted, thick] ($(x1)+(+2mm,-5mm)$)--(x1) ;
        \draw[dotted, thick] ($(x2)+(-2mm,-5mm)$)--(x2) ;
        \draw[dotted, thick] ($(x2)+(+2mm,-5mm)$)--(x2) ;
        
        \node at (alphap) [below left = 3.5mm and -0.7cm of alphap, text width=1cm] {$\genericMap(\alpha)$};
        \node at (betap) [below left = 3.5mm and -0.7cm of betap, text width=1cm] {$\genericMap(\beta)$};
        \node at (gammap) [below left = 3.5mm and -0.7cm of gammap, text width=1cm] {$\genericMap(\gamma)$};

        \node[above right =2.3cm and 2.1cm of eta, text width =3cm] (){Source pattern};
        \node[right = 2.8cm of arrow1, text width=1cm]  {$\xrightarrow{\substitution}$};
        \node[ right = 4.5cm of eta]  (phanta) {};

        \coordinate (alphas) at ($(phanta) + (1.7cm,-1.5cm)$);
        \coordinate (betas) at  ($(phanta) + (3cm,-1.5cm)$) ;
        \coordinate (gammas) at ($(phanta) + (4.3cm,-1.5cm)$);

        \fill[color=blue!20] ($(phanta) + (3cm,18mm)$) -- ($(phanta) + (4.7cm,-1.2cm)$) -- ($(phanta) + (1.3cm,-1.2cm)$);		
        \draw[color=blue, thick] ($(phanta) + (3cm,18mm)$) -- ($(phanta) + (4.7cm,-1.2cm)$) -- ($(phanta) + (1.3cm,-1.2cm)$) -- ($(phanta) + (3cm,18mm)$);		
        \fill[color=magenta!20] (alphas) -- ($(phanta) + (2.2cm,-2.4cm)$) -- ($(phanta) + (1.2cm,-2.4cm)$);			
        \draw[color=magenta, thick] (alphas) -- ($(phanta) + (2.2cm,-2.4cm)$) -- ($(phanta) + (1.2cm,-2.4cm)$) -- (alphas);	
        \fill[color=green!20] (betas)-- ($(phanta) + (3.5cm,-2.4cm)$) -- ($(phanta) + (2.5cm,-2.4cm)$);		
        \draw[color=green, thick] (betas) -- ($(phanta) + (3.5cm,-2.4cm)$) -- ($(phanta) + (2.5cm,-2.4cm)$) -- (betas);							
        \fill[color=orange!20] (gammas) -- ($(phanta) + (4.8cm,-2.4cm)$) -- ($(phanta) + (3.8cm,-2.4cm)$);
        \draw[color=orange, thick] (gammas) -- ($(phanta) + (4.8cm,-2.4cm)$) -- ($(phanta) + (3.8cm,-2.4cm)$) -- (gammas) ;
    
        \draw[dotted, thick, text width=1cm] ($(phanta) + (3cm,18mm)$) -- ($(phanta)+(3cm,+21mm)$)  node[right , xshift = 1pt] {$q_{10}$} ;
        \draw[dotted, thick] ($(phanta)+(1.4cm,-27mm)$)--($(phanta) + (1.5cm,-2.4cm)$) ;
        \draw[dotted, thick] ($(phanta)+(+20mm,-27mm)$)--($(phanta) + (1.9cm,-2.4cm)$) ;
        \draw[dotted, thick] ($(phanta)+(+25mm,-27mm)$)--($(phanta) + (2.8cm,-2.4cm)$) ;
        \draw[dotted, thick] ($(phanta)+(+30mm,-27mm)$)--($(phanta) + (3cm,-2.4cm)$) ;
        \draw[dotted, thick] ($(phanta)+(+35mm,-27mm)$)--($(phanta) + (3.2cm,-2.4cm)$) ;
        \draw[dotted, thick] ($(phanta)+(4cm,-27mm)$)--($(phanta) + (4.1cm,-2.4cm)$) ;
        \draw[dotted, thick] ($(phanta)+(+46mm,-27mm)$)--($(phanta) + (4.5cm,-2.4cm)$) ;
    
        \draw[text width=1cm] ($(alphas)+(0,3mm)$)--(alphas) node[above right, , xshift = 1pt] {$q_{7}$}  ;
        \draw[text width=1cm] ($(betas)+(0,3mm)$)--(betas) node[above right, xshift = 1pt] {$q_{8}$} ;
        \draw[text width=1cm] ($(gammas)+(0,3mm)$)--(gammas) node[above right, xshift = 1pt] {$q_{9}$} ;
    
        \node[fill=red!20, draw = red,  above right=5mm and 27mm of phanta,  regular polygon,regular polygon sides=3, text width = 4mm, label={[text width=1cm, xshift=5pt, yshift=-8mm]$\substitution(x)$}] (x1) {};
        
        \node[fill=red!20,draw = red, regular polygon,regular polygon sides=3, text width = 4mm, below left=8mm and 4mm of x1, label={[text width=1cm, xshift=5pt, yshift=-8mm]$\substitution(x)$}] (x2) {};
        \draw[text width=1cm, dotted, thick] (x1) -- ($(x1)+(-2mm,-5mm)$) node[left, xshift = 21pt] {\color{red}$q_1$};
        \draw[text width=1cm, dotted, thick] (x1) -- ($(x1)+(+2mm,-5mm)$) node[right, xshift = 1pt] {\color{red}$q_2$};
        \draw[text width=1cm, dotted, thick] (x1) -- ($(x1)+(0mm,+8mm)$) node[right, xshift = 1pt] {\color{red}$q_3$};
        \draw[text width=1cm, dotted, thick] (x2) -- ($(x2)+(-2mm,-5mm)$) node[left, xshift = 20pt] {\color{red}$q_4$};
        \draw[text width=1cm, dotted, thick] (x2) -- ($(x2)+(+2mm,-5mm)$) node[right, xshift = 1pt] {\color{red}$q_5$};
        \draw[text width=1cm, dotted, thick] (x2) -- ($(x2)+(0mm,+8mm)$) node[left, xshift = 18pt] {\color{red}$q_6$};
            
        \node at (alphas) [text width=1cm, below left  = 6.5mm and -5.5mm of alphas] {\scriptsize $\substitution(\genericMap(\alpha))$};
        \node at (betas) [text width=1cm, below left  = 6.5mm and -5.5mm of betas] {\scriptsize $\substitution(\genericMap(\beta))$};
        \node at (gammas) [text width=1cm, below left  = 6.5mm and -5.5mm of gammas] {\scriptsize $\substitution(\genericMap(\gamma))$};
        \node[above right =2.5cm and 13mm of phanta, text width =4cm] () {Annotated Source tree};
    \end{tikzpicture}}
    \caption{ The source tree is annotated with states of $\automaton_{\sourceLanguage}$. We want to simulate the effect of running $\automaton_{\sourceLanguage}$ on the source tree in the seed tree itself. This means, when processing $\eta$, it should guess and validate the potential transformation of the tuple $(q_7, q_8, q_9)$ to $q_{10}$ by a context (colored blue) in the source tree that matches $\sourceMap(\eta)$, and simultaneously make sure that there is a substitution for $x$ that preserves the transformations $(q_1, q_2)$ to $q_3 $ as well as $(q_4, q_5)$ to $q_6 $. An alternating tree automaton on the seed tree can achieve this.}\label{fig:proofOutline}
\end{figure}
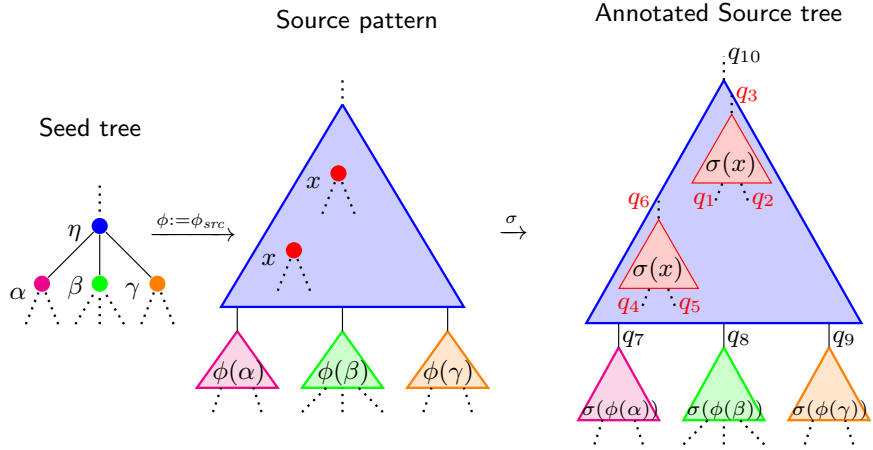
\clearpage

To solve the \metaTypecheckingProblemReformulationFull, we need to check for the existence of a seed tree $s$ giving rise to an atomic transformation $T_s$, a source tree $t_\src$, a target tree $t_\tgt$ and a substitution $\substitution$ satisfying two constraints. First, the source tree $t_\src$ should be in $\sourceLanguage$ and second, the target tree $t_\tgt$ should be in $\targetLanguage$. The source tree is obtained from the seed tree by first applying the homomorphism $\sourceMap$ and then applying the substitution $\substitution$. Suppose a node $\eta$ in the seed tree is transformed into a context as shown in the middle of Fig.~\ref{fig:proofOutline}. Suppose a variable $x$ is substituted with the tree $\substitution(x)$ as shown in the right side of Fig.~\ref{fig:proofOutline}. Consider a run of $\automaton_{\sourceLanguage}$ on the annotated source tree, in particular the state transformation ($(q_7,q_8,q_9)$ to $(q_{10})$) when the run parses the portion of the source tree that comes from the node $\eta$ of the seed tree. The idea of our decision procedure is to design a tree automaton (we will call this the \emph{\fsa} in the next two paragraphs) that parses the seed tree, and when it traverses the node $\eta$, it simulates $\automaton_{\sourceLanguage}$'s state transformation $(q_7,q_8,q_9)$ to $(q_{10})$. To determine that $(q_7,q_8,q_9)$ is transformed to $q_{10}$, the \fsa\ needs to know that $(q_1,q_2)$ is transformed to $q_3$ by $\substitution(x)$.

The \fsa\ doesn't know what is $\substitution$ (in fact, our goal is to check if a suitable $\substitution$ exists). There might be infinitely many substitutions $\substitution$ out of which one might work; but our \fsa\ is supposed to be a finite state automaton, incapable of picking one choice from infinitely many. To work around this, we observe that the actual tree $\substitution(x)$ is not important; the important thing is that it transforms $(q_1,q_2)$ to $(q_3)$. Any other $\substitution'(x)$ which does the same transformation will work equally well in place of $\substitution(x)$. The number of such transformations is finite (since they are transformations on a finite set of states) and the \fsa\ only needs to check if one of these finitely many transformations work.

One thing we ignored in the above explanation is that the variable $x$ may occur multiple times in the source pattern. In the run of $\automaton_{\sourceLanguage}$ on the source tree, one occurrence of $\substitution(x)$ may encounter $(q_1,q_2)$ (which is transformed to $q_3$). Some other occurrence of $\substitution(x)$ may encounter some other pair $(q_4,q_5)$ (which is transformed to some other state, say $q_6$). What we need is a substitution that transforms $(q_1,q_2)$ to $q_3$ and also transforms $(q_4,q_5)$ to $q_6$ (and does other transformations too, if there are more occurrences of $x$ in the source pattern). To capture all these transformations simultaneously, we introduce the concept of relational lifts that is formally defined in the next section. The \fsa\ nondeterministically picks out one such possible set of transformations to check that the source tree is in $\sourceLanguage$. It similarly needs to check that the target tree is in $\targetLanguage$. Alternating tree automata offer a convenient way to achieve all these tasks of the \fsa. The proof idea is to design an alternating tree automaton whose nonemptiness is equivalent to the existence of a seed tree that satisfies the two required constraints.


\section{Alternating Tree Automata and Relational Lifts}\label{sec:relationalLifts}

We introduce concepts that will be essential for the final proof of the theorem.

\begin{definition}[Alternating Top-down Tree Automata]\label{def:alternatingTreeAutomaton}
    An \emph{Alternating Top-down Tree Automaton} or \emph{ATTA} over the alphabet $\alphabet$ is a tuple $\automaton := (\automatonStateSet, \alphabet, \automatonTransitionsSet, \automatonInitialStatesSet)$, where $\automatonStateSet$ is a finite set of \emph{states}, $\automatonInitialStatesSet \subseteq \automatonStateSet$ is a set of initial states, and $\automatonTransitionsSet : \automatonStateSet \times \alphabet \to \B^+(\automatonStateSet \times \N)$ such that $\automatonTransitionsSet(q, a/k) \in \B^+(\automatonStateSet \times [k])$, for all $q \in \automatonStateSet, a/k \in \alphabet$. Here, for any set $X$, $\B^+(X)$ denotes the set of all positive boolean combinations of elements in $X$.
\end{definition}

The semantics of ATTA are defined using runs, which are themselves trees satisfying some conditions.
\begin{definition}[Run, Acceptance for ATTA]\label{def:acceptanceByAtta}
    Given a tree $t \in \trees{\alphabet}$ and an ATTA $\automaton$ over $\alphabet$, we define a \emph{run} of $\automaton$ on $t$ to be a tree $\run$ over $\automatonStateSet \times \N^*$ such that $\run(\e) = (q, \e)$ for some state $q$ and every position $u \in \Positions{\run}$ satisfies the following condition: if $\run(u) = (q, x), t(x) = a/k, \automatonTransitionsSet(q ,a) = \phi$, then there is a subset $S := \set{(q_1, i_1), \dots (q_n, i_n)} \subseteq \automatonStateSet \times [k]$ such that $S \models \phi$, the successor positions of $u$ in $\run$ are $\set{u1, \dots, un}$, and $\run(uj) = (q_j, xi_j)$ for all $j \in [n]$.

    A run is \emph{successful} if $\run(\e) = (q, \e)$ for some initial state $q \in \automatonInitialStatesSet$. A tree $t$ is \emph{accepted} by ATTA $\automaton$ if there exists at least one successful run of $\automaton$ on $t$. The set of all trees accepted by ATTA $\automaton$ is called the tree language \emph{recognised} by $\automaton$, and denoted by $\languageOf{\automaton} \subseteq \trees{\alphabet}$. A tree language $\genericLanguage \subseteq \trees{\alphabet}$ is called \emph{regular} if there exists an ATTA $\automaton$ such that $\genericLanguage = \languageOf{\automaton}$. The set of all regular tree languages over $\alphabet$ is denoted by $\regularOver{\alphabet} \subseteq 2^{\trees{\alphabet}}$.
\end{definition}

Suppose $\automaton$ is a word automaton with set of states $\automatonStateSet$ and transition relation given by $\automatonTransitionsSet : \alphabet \to \relations(\automatonStateSet,2)$, we can extend it to $\widetilde{\automatonTransitionsSet} : \alphabet^* \to \relations(\automatonStateSet,2)$ such that $(q_1,q_2) \in \widetilde{\automatonTransitionsSet}(w)$ iff $\automaton$ has a run on the word $w$ starting from $q_1$ and ending at $q_2$. The following defines a similar extension for NBTTS.

\begin{definition}[Extended Transition Relation(\extendedTransitionRelation)]\label{def:extendedTransitionRelation}
    Let $\transitionSystem := (\automatonStateSet, \alphabet, \automatonTransitionsSet)$ be a NBTTS. We extend the transition function $\automatonTransitionsSet$ to a larger function $\widetilde{\automatonTransitionsSet} : \contexts{\alphabet}{} \to \relations(\automatonStateSet)$ called the \emph{Extended Transition Relation} or \emph{\extendedTransitionRelation} as follows. Let $c \in \contexts{\alphabet}{k}$. For a tuple $\tuple{q} := (q_1, q_2, \dots, q_k)$ of states, define the NBTTS $\transitionSystem_{\tuple{q}} := (\automatonStateSet, \alphabet \uplus \set{\square_1/0, \square_2/0, \dots, \square_k/0}, \automatonTransitionsSet \uplus \set{(\square_i, \{q_i\}) \mid i \in [k]})$. Let us define the set of results to be the set given by $\Results(\tuple{q}) := \set{(\tuple{q}, q) \in \automatonStateSet^{k+1} \mid q \mbox{ is the result of a run of $\transitionSystem_{\tuple{q}}$ on $c$.}}$. Finally, we also define the set $\widetilde{\automatonTransitionsSet}(c) := \bigcup_{\tuple{q} \in \automatonStateSet^k} \Results(\tuple{q})$. We also refer to the relation $\widetilde{\automatonTransitionsSet}(c)$ as the \emph{relation induced} by the context $c$ on the states of the NBTTS $\transitionSystem$. Note that $\widetilde{\automatonTransitionsSet}(a) = \automatonTransitionsSet(a),~\forall a \in \alphabet$.
\end{definition}

Suppose a word $w_1$ is extended by appending another word $w_2$ to its end. Now, $\widetilde{\automatonTransitionsSet}(w_1 \cdot w_2)$ is obtained by simply composing the relations $\widetilde{\automatonTransitionsSet}(w_1)$ and $\widetilde{\automatonTransitionsSet}(w_2)$. For trees, it is more complicated: a tree $t_1$ can have multiple ``end points'' where it can be extended by appending other trees. We will need a more complicated way of composing relations.

\begin{definition}[Relational Lifts]\label{def:relationalLift}
    Let $X$ be a set and $n \in \N_0$ be a nonnegative integer. Let $\operator \subseteq X^{n+1}$ be a relation. We define the \emph{relational lift} of $\operator$ as a new relation $\widetilde{\operator} \subseteq {\relations(X)}^{n+1}$ given as follows: For $R_1, \dots, R_n, R \in \relations(X)$, we say that $(R_1, R_2, \dots, R_n, R) \in \widetilde{\operator}$ if and only if
    \begin{itemize}
        \item $R \in \relations\left(X, 1+\sum_{j=1}^n i_j\right)$, where $i_j \in \N_0$ for all $j \in [n]$ are such that $R_j \in \relations(X, i_j+1)$.
        \item For all $j \in [n]$, for all $\tuple{x}_j \in X^{i_j}$, we have $(\tuple{x}_1, \tuple{x}_2, \dots, \tuple{x}_n, y) \in R$ iff there exist $y_1, y_2, \dots, y_n \in X$ such that $(\tuple{x}_j, y_j) \in R_j$ and $(y_1, y_2, \dots, y_n, y) \in \operator$.
    \end{itemize}    
\end{definition}

We can think about relational lift as a procedure to compose arbitrarily many relations. Another, more abstract way to think about it is as a way of ``lifting'' an operator over a set to an operator over relations over that same set. In case of automata running on words, for any $a \in \alphabet$ and any $w_1 \in \alphabet^*$, if we set $\operator := \automatonTransitionsSet(a)$, then $(\widetilde{\automatonTransitionsSet}(w_1), \widetilde{\automatonTransitionsSet}(w_1 \cdot a)) \in \widetilde{\operator}$. In case of NBTTS running on contexts, suppose $c_1, \ldots, c_k$ are the children of a node labeled by the letter $a/k$ and we set $\operator := \automatonTransitionsSet(a) \subseteq \automatonStateSet^{k+1}$, then $(\widetilde{\automatonTransitionsSet}(c_1), \ldots, \widetilde{\automatonTransitionsSet}(c_k), R) \in \widetilde{\operator}$ implies that $R = \widetilde{\automatonTransitionsSet}(a(c_1, \ldots, c_k))$.

The following problem and its complexity is an important intermediate technical lemma used later.
\begin{problemDefinition}
    \problemtitle{\relationRealisabilityFull}
    \probleminput{NBTTS $\transitionSystem := (\automatonStateSet, \alphabet, \automatonTransitionsSet)$, relation $R \in \relations(\automatonStateSet)$, NBTA $\D$.}
    \problemquestion{Is $R$ \emph{realisable} by $\transitionSystem$ under $\D$ i.e., Does there exist $c \in \contexts{\alphabet}{} \cap \languageOf{\D}$ such that $\widetilde{\automatonTransitionsSet}(c) = R$?}
\end{problemDefinition}

  As a matter of nomenclature, if the above condition holds for a given $\transitionSystem, R, \D$, then we say that the relation $R$ is \emph{realisable} by $\transitionSystem$ under $\D$.

\begin{theorem}\label{thm:relationRealisabilityExptc}
The \textsc{\relationRealisabilityFull}\ is \exptc.
\end{theorem}

\begin{proof}
    (Lower Bound) We show a reduction from \intersectionNonemptinessFull\ to \relationRealisabilityFull. We recall the problem \intersectionNonemptinessFull\ as:

    \begin{problemDefinition}
        \problemtitle{\intersectionNonemptinessFull}
        \probleminput{DBTAs $\automaton_1, \automaton_2, \dots, \automaton_n$.}
        \problemquestion{Do we have that $\bigcap_{i=1}^n L(\automaton_i) \neq \emptyset$?}
    \end{problemDefinition}

    We know from~\cite{tata} that \intersectionNonemptinessFull\ is \exptc. Let $\alphabet$ be a given finite alphabet, and let $\automaton_1, \dots, \automaton_n$ be an instance of \intersectionNonemptinessFull, where $\automaton_i := (Q_i, \alphabet, \Transitions_i, F_i)$ are DBTAs. If $F_i = \emptyset$ for some $i \in [n]$, we must have $L(\automaton_i) = \emptyset$ and hence the intersection $\bigcap_{i=1}^n L(\automaton_i) = \emptyset$. Hence we assume that $F_i \neq \emptyset$ for all $i \in [n]$.

    Consider NBTA $\automatonB := \left((\biguplus_{i=1}^n Q_i) \uplus \{\top\}, \seedAlphabet := \alphabet \uplus \{\#/2\}, \Transitions, \{\top\} \right)$, with $\Transitions : \alphabet \uplus \set{\#/2} \to \relations(Q)$ is defined as
    \[ \twopartfunc{\Transitions(a)}{\biguplus_{i=1}^n \Transitions_i(a)}{a \in \alphabet}{\{(q, q', \top) \mid q, q' \in F_i \mbox{ for some } i \in [n] \}}{a = \#} \]
    Let $R \in \relations(Q, 2)$ be given by $R := \set{(f_i, \top) \mid i \in [n]}$, where each $f_i$ is an arbitrarily chosen but fixed final state in $F_i$ for each $i \in [n]$. Also consider DBTA \shape\ which recognises all contexts of the form $\#(t, \square)$ where $t \in \trees{\alphabet}$. Consider the instance of the realisability problem given by $(\automatonB, R, \shape)$.

    We can see that $R$ is realisable by $\automatonB \cap \shape$ if and only if there exists a tree $t \in L(\automaton_1) \cap \dots \cap L(\automaton_n)$. Clearly, $R$ is realisable by $\automatonB \cap \shape$ $\iff$ there exists a context $c := \#(t, \square)$ which induces the relation $R$ on the states of $\automatonB$ under $\shape$ $\iff$ for each $i \in [n]$ and the labelling of the hole $\square$ with $f_i$, the result of a run of $\automatonB$ on $t$ yields a final state $f_i' \in F_i$ $\iff$ $t$ is recognised by $\automaton_i$ for each $i \in [n]$ $\iff$ $t \in \bigcap_{i=1}^n L(\automaton_i)$. Since we know that \intersectionNonemptinessFull\ is \expth, we conclude that the realisability problem is also \expth.

    \textbf{(Upper Bound)} Let NBTTS $\transitionSystem := (Q, \alphabet, \Transitions)$ be a transition system, relation $R \in \relations(Q, k+1) \subseteq \relations(Q)$, and NBTA $\D$ be a given instance of \relationRealisabilityFull. We reduce the problem of finding a context $c \in \contexts{\alphabet}{k} \cap L(\D)$ such that $\widetilde{\Transitions}(c) := R$ to the problem of NBTA nonemptiness. We will define an NBTA $\automaton$ with relations over $Q$ as states such that for any context $c$, the result of a run of $\automaton$ on $c$ is precisely the relation induced by $c$. More formally, we define a NBTA $\automaton$ as: $\automaton := (\bigcup_{i=1}^{k+1} \relations(Q, i), \alphabet \uplus \{\square_i/0 \mid i \in [k]\}, \Transitions', \{R\})$, where for all $a/p \in \alphabet$ we set $\Transitions'(a) := \widetilde{\Transitions(a)}$, the relational lift of $\Transitions(a)$. We also set $\Transitions'(\square_i) := Q$ for all $i \in [k]$.

    A straightforward induction on the height of a tree gives us that for a context $c \in \contexts{\alphabet}{k} \cap L(\D)$, we have $\widetilde{\Transitions}(c) = R$ if and only if $c \in L(\automaton) \cap L(\D)$. Hence to find a context which `realises' the relation $R$, all we need to do is check $\automaton \cap \D$ for nonemptiness. Let automaton $\D$ have $m$ states. We know from~\cite{tata} that checking this NBTA for nonemptiness can be performed in time which is a polynomial in $|Q|$. Since we have that $m|Q| := m|\bigcup_{i=1}^{k+1} \relations(Q, i)| = m\sum_{i=1}^{k+1}|\relations(Q, i)| = m\sum_{i=1}^{k+1}2^{|Q|^i}$, we see that given $\transitionSystem, R, \D$, the problem of \relationRealisabilityFull\ can be solved in \expt.
\end{proof}

In the proof for the above result, we also get that for each relation which we are checking for realisability, we can also get, in \expt, a context over the alphabet which realises it(if it is realisable). In such a case, we will say that the relation $R$ is \emph{realised} by the context $c$ on the NBTTS $\transitionSystem$, under $\D$. We also say that $c$ \emph{induces} $R$ on $\transitionSystem$ under $\D$.  For a given pair $(\transitionSystem, \D)$ with NBTTS $\transitionSystem$ and NBTA $\D$, we define the set $\realisableRelations{\transitionSystem}{\D} := \{(R, c) \mid \mbox{ $R$ is realised by $c$ on $\transitionSystem$ under $\D$}\}$.


\section{Deciding Transformation Reachability}\label{sec:proofDetails}

Let $\fullTransform := \left((\seedLanguage, \sourceMap, \targetMap, \guardFunction), \sourceLanguage, \targetLanguage \right)$ be a given instance of \metaTypecheckingProblemReformulationFull, where
\begin{itemize}
    \item $\seedLanguage$ is the seed language over $\seedAlphabet$ given by the NBTA $\automaton := (P, \seedAlphabet, \Theta, E)$.
    \item $\sourceMap, \targetMap : \trees{\seedAlphabet} \to \trees{\alphabet \uplus \variables}$ are tree homomorphisms.
    \item $\guardFunction$ is a guard over $\alphabet$ such that $\variables(\sourceMap) \cup \variables(\targetMap) \subseteq \dom(\guardFunction)$, where we define(by abuse of notation) $\variables(\sourceMap) := \bigcup_{\gamma \in \seedAlphabet} \variables(\sourceMap(\gamma))$ and also $\variables(\targetMap) := \bigcup_{\gamma \in \seedAlphabet} \variables(\targetMap(\gamma))$.
    \item $\sourceLanguage, \targetLanguage$ are the source and target languages respectively, which are regular tree languages given by the NBTAs $\automatonB_i := (Q_i, \alphabet, \Transitions_i, F_i)$ for $i \in \set{\src, \tgt}$.
\end{itemize}

We see that solving the \metaTypecheckingProblemReformulationFull\ is equivalent to finding whether there exists $\tau \in \seedLanguage$ and a substitution $\substitution$ such that $\substitution(\phi_i(\tau)) \in \genericLanguage_i$ for $i \in \set{\src, \tgt}$. Let $J \subseteq \seedLanguage$ be given by $J := \set{\tau \in \seedLanguage \mid \exists\substitution\mbox{ such that }\substitution(\phi_i(\tau)) \in \genericLanguage_i, i \in \set{\src, \tgt}}$. We will construct a number of ATTAs such that the union of the languages they recognise will be equal to $J$. This will imply that solving \metaTypecheckingProblemReformulationFull\ will be a matter of checking each of the correspondingly constructed ATTAs for nonemptiness.

For each $x \in \variables(\sourceMap) \setminus \variables(\targetMap)$, using Theorem~\ref{thm:relationRealisabilityExptc} we compute a priori, the set {of relations} $\realisableRelations{\sourceAutomatonB}{\guardOf{x}}$ of all the relations and corresponding contexts realisable by $\sourceAutomatonB$ under $\guardOf{x}$. Similarly for each variable $y \in \variables(\targetMap) \setminus \variables(\sourceMap)$, we compute $\realisableRelations{\targetAutomatonB}{\guardOf{y}}$.
For all variables $z \in \variables(\sourceMap) \cap \variables(\targetMap)$, we compute the set
\[ \realisableRelations{\sourceAutomatonB \cap \targetAutomatonB}{\guardOf{z}} \subseteq \relations(Q_\src \times Q_\tgt) \times \contexts{\alphabet}{}. \]

From this set, we obtain another set $\realisableRelations{\sourceAutomatonB, \targetAutomatonB}{\guardOf{z}}$ of pairs of relations that are \emph{simultaneously} realisable by $\sourceAutomatonB$ and $\targetAutomatonB$ along with a context $c \in L(\guardOf{z})$ that realises both of them. The idea  is that given an element $((q_\src^{(1)}, q_\tgt^{(1)}), (q_\src^{(2)}, q_\tgt^{(2)}), \dots, (q_\src^{(k+1)}, q_\tgt^{(k+1)})) \in R$, we `split' it into two components $(q_\src^{(1)}, q_\src^{(2)}, \dots, q_\src^{(k+1)})$ and $(q_\tgt^{(1)}, q_\tgt^{(2)}, \dots, q_\tgt^{(k+1)})$. We do this for all elements in $R$ to get a pair of relations $R_\src, R_\tgt$ over $Q_\src, Q_\tgt$ respectively. Then, we define $\realisableRelations{\sourceAutomatonB, \targetAutomatonB}{\guardOf{z}}$ as
\[  := \{((R_\src, R_\tgt), c) \mid (R, c) \in \realisableRelations{\sourceAutomatonB \cap \targetAutomatonB}{\guardOf{z}}\}. \]
In other words, we have $\realisableRelations{\sourceAutomatonB, \targetAutomatonB}{\guardOf{z}} := \{((R_\src, R_\tgt), c) \mid c \in L(\guardOf{z}), \widetilde{\Transitions_i}(c) = R_i, i \in \{\src, \tgt\}\}$.

We fix a pair of choice functions of realisable relations $R_\src, R_\tgt$ over $\variables(\sourceMap), \variables(\targetMap)$ respectively such that
\[ \twopartfunc{R_\src(z)}{R \in \pi_1[\realisableRelations{\automatonB_s}{\guardOf{z}}]}{z \in \variables(\sourceMap) \setminus \variables(\targetMap)}{R' \in \pi_1[\pi_1[\realisableRelations{\sourceAutomatonB, \targetAutomatonB}{\guardOf{z}}]]}{z \in \variables(\sourceMap) \cap \variables(\targetMap)} \]
\[ \twopartfunc{R_\tgt(z)}{R \in \pi_1[\realisableRelations{\targetAutomatonB}{\guardOf{z}}]}{z \in \variables(\targetMap) \setminus \variables(\sourceMap)}{R'' \in \pi_2[\pi_1[\realisableRelations{\sourceAutomatonB, \targetAutomatonB}{\guardOf{z}}]]}{z \in \variables(\sourceMap) \cap \variables(\targetMap)} \]
and $(R', R'') \in \pi_1[\realisableRelations{\sourceAutomatonB, \targetAutomatonB}{\guardOf{z}}]$ for all $z \in \variables(\sourceMap) \cap \variables(\targetMap)$. Note that this also provides us with corresponding contexts which realise the chosen relations.

Armed with this, we compute the \extendedTransitionRelation{} s\ $\widetilde{\Transitions_i}(\phi_i(\gamma))$ for all $\gamma \in \seedAlphabet$ and $i \in \{\src, \tgt\}$ by simply considering each $x \in \variables(\fullTransform)$ as a separate letter with the `transition relation' given by the corresponding $R_i(x)$ and running the corresponding automaton on it.

\begin{lemma}\label{lem:constructionOfAttaZ}
    For each pair of choice functions $R_\src, R_\tgt$, there exists an ATTA $\Z_{R_\src, R_\tgt}$ of size $\bigo(\max{\abs{Q_\src}, \abs{Q_\tgt}})$ such that
    \[ L(\Z_{R_\src, R_\tgt}) = \set{\tau \mid \substitution(\sourceMap(\tau)) \in \sourceLanguage, \substitution(\targetMap(\tau)) \in \targetLanguage} \]
    where $\substitution$ is a function which assigns to each variable $x$, a context that induces the relation $R_i(x)$ for $i \in \{\src, \tgt\}$ and satisfies the guard $\guardFunction$.
\end{lemma}

We first give an explicit construction of the ATTA $\Z_{R_\src, R_\tgt}$. To describe it succinctly, we define a few related concepts.

\begin{definition}[Destination Function]\label{def:destinationFunction}
Let $\seedLanguage, \genericLanguage$ be tree languages over finite ranked alphabets $\seedAlphabet, \alphabet$ respectively and let $\genericMap : \trees{\seedAlphabet} \to \trees{\alphabet}$ be a tree homomorphism. For each letter $\gamma/k \in \seedAlphabet$ with $m := \arity(\genericMap(\gamma))$, we define a function $\destinationFunction{\genericMap}{\gamma} : [k] \to 2^{[m]}$ as follows. Since $\genericMap(\gamma)$ is a context of arity $m$, we label the holes of $\genericMap(\gamma)$ by the set $[m]$ according to the precedence in the inorder traversal of $\genericMap(\gamma)$. Then, $\destinationFunction{\genericMap}{\gamma}(i) := \{j \in [m] \mid \mbox{ $j^{th}$ hole of $\genericMap(\gamma)$ is labelled $u_i$}\}$ for all $i \in [k]$. 
\end{definition}

We illustrate this with an example. Let $\gamma/5 \in \seedAlphabet$ be a letter and let $\genericMap(\gamma/5)$ be given as in Figure~\ref{fig:destinationFunctionExample}. The numbers in blue denote the order of the vertex in the inorder traversal of $\genericMap(\gamma)$.

\begin{figure}[!ht]
    \begin{center}
    \scalebox{1}{\begin{tikzpicture}[level distance=3em]
        \begin{scope}[yshift=-2cm]
            \node [sibling distance=1em] (root) {$\gamma$}
            child { node {$u_{1}$} }
            child { node {$u_{2}$} }
            child { node {$u_{3}$} }
            child { node {$u_{4}$} }
            child { node {$u_{5}$} };
            \node at ([xshift=3cm]root) {$\xrightarrow{\genericMap}$};
        \end{scope}
        \begin{scope}[xshift=6cm]
            \node {$a$}
            child[grow=south west] { 
            node{$x$}
            child { node{$u_1\blue{(1)}$} }
            child { node{$u_3\blue{(2)}$} }
            }
            child[grow=south] {
            node {$b$}
            child[grow=south] {
                node{$y$}
                child { node{$u_1\blue{(3)}$} }
                child {
                node{$x$}
                child { node{$u_2\blue{(4)}$} }
                child { node{$u_4\blue{(5)}$} }
                }
                child { node{$u_1\blue{(6)}$} }
            }
            child[grow=south east] {
            node{$a$}
            child[grow=south east] { node{$c$} }
            child[grow=east] { node{$u_3\blue{(7)}$}}}};
        \end{scope}
    \end{tikzpicture}}
    \caption{Example computation of $\destinationFunction{\genericMap}{\gamma}$}\label{fig:destinationFunctionExample}
    \end{center}
\end{figure}

Then $\destinationFunction{\genericMap}{\gamma} : [5] \to 2^{[7]}$ would be given by

\begin{tabular}{|c||c|c|c|c|c|}
\hline
$i$ & 1 & 2 & 3 & 4 & 5 \\\hline
$\destinationFunction{\genericMap}{\gamma}(i)$ & $\set{1, 3, 6}$ & $\set{4}$ & $\set{2, 7}$ & $\set{5}$ & $\emptyset$ \\\hline
\end{tabular}

Note that as demonstrated by this example, in each case we must necessarily have that the set $\set{\destinationFunction{\genericMap}{\gamma}(i) \mid i \in [k]}$ is a partition of $[\arity(\genericMap(\gamma))]$. Intuitively, we also consider the function $\destinationFunction{\genericMap}{\gamma}$ as a kind of ``inverse'' of the origin function as mentioned in~\cite{BojanOrigin14}.

Coming to the construction of the ATTA, for each pair of choice functions $R_\src, R_\tgt$, we define ATTA $\Z_{R_\src, R_\tgt}$ given by
\[ \Z_{R_\src, R_\tgt} := \left((P \biguplus Q_\src \biguplus Q_\tgt \biguplus \{\start\}), \seedAlphabet, I := \{\start\}, \Transitions^\prime\right) \]
where we define $\Transitions'$ as follows. For a fixed letter $\gamma \in \seedAlphabet$,
\begin{itemize}
    \item For a state $p \in P$, we set
    \[ \Transitions'(p, \gamma/k) := \bigvee_{(p_1, p_2, \dots, p_k, p) \in \Theta(\gamma)}(p_1, 1) \land (p_2, 2) \land \dots \land (p_k, k) \]
    This transition basically ensures that the ATTA $\Z_{R_\src, R_\tgt}$ simulates $\automaton$.
    \item For a state $q \in Q_\src$, we set
    \[ \Transitions'(q, \gamma/k) := \bigvee_{(q_1, \dots, q_m, q) \in \widetilde{\Transitions_\src}(\sourceMap(\gamma))} \bigwedge_{i \in [k]} \left[\bigwedge_{j \in \destinationFunction{\sourceMap}{i}}(q_j, i)\right] \]
    Intuitively, this transition simulates the running of $\sourceAutomatonB$ on $\sourceMap(\tau)$.
    \item Similarly, for a state $q \in Q_\tgt$, we set
    \[ \Transitions'(q, \gamma/k) := \bigvee_{(q_1, \dots, q_m, q) \in \widetilde{\Transitions_\tgt}(\targetMap(\gamma))} \bigwedge_{i \in [k]} \left[\bigwedge_{j \in \destinationFunction{\targetMap}{i}}(q_j, i)\right] \]
    Intuitively, this transition simulates the running of $\targetAutomatonB$ on $\targetMap(\tau)$.
    \item We also have an $\e$-transition $\Transitions^\prime(\start, \e/1) := \bigvee_{p \in E, q_\src \in F_\src, q_\tgt \in F_\tgt}((p, 1) \land (q_\src, 1) \land (q_\tgt, 1))$
\end{itemize}

We illustrate this with an example. Let $q \in Q_\src$, and $\gamma/5 \in \seedAlphabet$, i.e., $\gamma$ is a letter with arity 5. Also let $\sourceMap(\gamma)$ be defined as in Figure~\ref{fig:destinationFunctionExample}. Note that we have already completely computed the extended transition relation $\widetilde{\Transitions_\src}(\sourceMap(\gamma))$ given the choice function $R_\src$. In such a case we define $\Transitions'(q, \gamma/5) :=$
\[
\bigvee_{\begin{matrix} (q_1, \dots, q_7, q)\\ \in \widetilde{\Transitions_\src}(\sourceMap(\gamma)) \end{matrix}} 
\begin{matrix}
    \left[ \underbrace{(q_1, 1) \land (q_3, 1) \land (q_6, 1)}_{\destinationFunction{\sourceMap}{1} = \{1, 3, 6\}}~~ \land \underbrace{(q_4, 2)}_{\destinationFunction{\sourceMap}{2} = \{4\}} \land~~ \underbrace{(q_2, 3) \land (q_7, 3)}_{\destinationFunction{\sourceMap}{3} = \{2, 7\}} ~~\land \underbrace{(q_5, 4)}_{\destinationFunction{\sourceMap}{4} = \{5\}} \right]
\end{matrix}
\]

\begin{proof}
    For the given pair of choice functions $R_\src, R_\tgt$, let $C_{R_\src, R_\tgt}$ be a function that assigns to each $x \in \variables(\sourceMap) \cup \variables(\targetMap)$, a context that induces the relation $R_i(x)$ for $i \in \{\src, \tgt\}$ and satisfies the guard $\guardFunction$.

    ($\Leftarrow$) Let $\substitution := C_{R_\src, R_\tgt}$ such that $\substitution(\sourceMap(\tau)) \in \sourceLanguage$ and $\substitution(\targetMap(\tau)) \in \targetLanguage$ for some $\tau \in \trees{\seedAlphabet}$. This implies that for an accepting run of $\sourceAutomatonB$ on $\substitution(\sourceMap(\tau))$, the result of that run is some final state $f_\src \in F_\src$. Similarly for an accepting run of $\targetAutomatonB$ on $\substitution(\targetMap(\tau))$, the result of that run is some final state $f_\tgt \in F_\tgt$. Also, since $\tau \in \seedLanguage$, there is an accepting run of $\automaton$ on $\tau$ such that the result of that run is $e \in E$. Then we can see that the set $\{e, f_\src, f_\tgt\} \models I$, the starting formula for $\Z_{R_\src, R_\tgt}$. From here the run of $\Z_{R_\src, R_\tgt}$ simulates the run of $\automaton$ on $\tau$, and the runs of $\automatonB_i$ on $\genericMap_i(\tau)$ for $i \in \{\src, \tgt\}$ conforming to the transition $C_{R_\src, R_\tgt}$. Since all of these are successful runs, the run of $\Z_{R_\src, R_\tgt}$ we described on $\tau$ is also a successful run, and hence $\tau \in L(\Z_{R_\src, R_\tgt})$.

    ($\Rightarrow$) Conversely, let $\tau \in L(\Z_{R_\src, R_\tgt})$. This implies that there is an accepting run of $\Z_{R_\src, R_\tgt}$ on $\tau$. Given the nature of the starting formula $I$, we must have an $e \in E, f_\src \in F_\src, f_\tgt \in F_\tgt$ such that $\{e, f_\src, f_\tgt\} \models I$. Since $\Z_{R_\src, R_\tgt}$ simulates $\automaton$ on $\tau$ and $e \in E$, we see that $\tau \in L(\automaton)$. Since $\Z_{R_\src, R_\tgt}$ also simulates $\automatonB_i$ on $\phi_i(\tau)$ conforming to $C_{R_\src, R_\tgt}$ for all $i \in \{\src, tgt\}$, we see that $\substitution(\genericMap_i(\tau)) \in \genericLanguage_i$ for $i \in \{\src, \tgt\}$. This completes the proof.    
\end{proof}

We know from~\cite{tata} that the emptiness problem for ATTA is \exptc\ in the size of the automaton. Note that this computation has been performed for a specific choice of $R_\src, R_\tgt$. Let $\mathcal{R}$ be the set of all pairs of choice functions. To complete the solution of the problem, we need to perform this computation for each pair of choice functions $(R_\src, R_\tgt) \in \mathcal{R}$. We know that $|\mathcal{R}|$ is bounded above by $\prod_{x \in \variables(\sourceMap)}2^{|Q_\src|^{\deg(x)}} \cdot \prod_{y \in \variables(\targetMap)}2^{|Q_\tgt|^{\deg(y)}}$. Let $\stateSize := \max\{|Q_\src|, |Q_\tgt|\}$. This gives an upper bound $|\R| \leq \prod_{x \in \variables(\sourceMap)}2^{\stateSize^{\deg(\fullTransform)}} \cdot \prod_{y \in \variables(\targetMap)}2^{\stateSize^{\deg(\fullTransform)}}$ which eventually simplifies to the bound $|\R| \leq {\left(2^{\stateSize^{\deg(\fullTransform)}}\right)}^{|\variables(\fullTransform)|}$. We also know from Lemma~\ref{lem:constructionOfAttaZ} that the size of the ATTA $\Z_{R_\src, R_\tgt}$ is bounded above by $\bigo(\stateSize)$ and hence each nonemptiness check requires time given by $2^{\bigo(\stateSize)}$. Hence the total time required for all the nonemptiness checks is given by:
\[ \mbox{Total running time  is in }  2^{\bigo(\stateSize)} \cdot 2^{|\variables(\fullTransform)| \cdot \stateSize^{\deg(\fullTransform)}}. \]
Since $|\variables(\fullTransform)| + \stateSize \leq |\mathsf{input}|$, on simplifying the above, we get that the running time of the algorithm we have provided is $2^{\mathsf{poly}(|\mathsf{input}|)}$, where $\mathsf{poly}$ is a polynomial of degree $\deg(\fullTransform)$. This completes the proof of Theorem~\ref{thm:metaTypecheckingUpperBound}.


\section{Conclusion}\label{sec:conclusion}

We have introduced a formal model of tree transformations which is highly expressive. While the equivalence checking of tree transformations is undecidable we show that  the type-checking problem is in 3-\expt\ by giving  an automata-theoretic algorithm. 

A future direction is to look at other variants of expressive tree transformations where the (source pattern, target pattern) pair is generated by more generic versions of tree transducers and study the decidability of type-checking problem.


\bibliography{references}

\begin{thebibliography}{10}

\bibitem{AMP2022}
C.~Aiswarya, Sahil Mhaskar, and M.~Praveen.
\newblock Checking regular invariance under tightly-controlled string
  modifications.
\newblock In {\em Developments in Language Theory: 26th International
  Conference, DLT 2022, Tampa, FL, USA, May 9–13, 2022, Proceedings}, page
  57–68, Berlin, Heidelberg, 2022. Springer-Verlag.
\newblock \href {https://doi.org/10.1007/978-3-031-05578-2_4}
  {\path{doi:10.1007/978-3-031-05578-2_4}}.

\bibitem{StreamingTreeTransd}
Rajeev Alur and Loris D'antoni.
\newblock Streaming tree transducers.
\newblock {\em J. ACM}, 64(5), aug 2017.
\newblock \href {https://doi.org/10.1145/3092842} {\path{doi:10.1145/3092842}}.

\bibitem{BojanOrigin14}
Miko{\l}aj Boja{\'{n}}czyk.
\newblock Transducers with origin information.
\newblock In Javier Esparza, Pierre Fraigniaud, Thore Husfeldt, and Elias
  Koutsoupias, editors, {\em Automata, Languages, and Programming}, pages
  26--37, Berlin, Heidelberg, 2014. Springer Berlin Heidelberg.

\bibitem{tata}
Hubert Comon, Max Dauchet, R{\'e}mi Gilleron, Florent Jacquemard, Denis Lugiez,
  Christof L{\"o}ding, Sophie Tison, and Marc Tommasi.
\newblock {\em {Tree Automata Techniques and Applications}}.
\newblock 2008.
\newblock URL: \url{https://inria.hal.science/hal-03367725}.

\bibitem{DVLM2015}
Loris D'antoni, Margus Veanes, Benjamin Livshits, and David Molnar.
\newblock Fast: A transducer-based language for tree manipulation.
\newblock {\em ACM Trans. Program. Lang. Syst.}, 38(1), oct 2015.
\newblock \href {https://doi.org/10.1145/2791292} {\path{doi:10.1145/2791292}}.

\bibitem{Engelfriet75}
Joost Engelfriet.
\newblock Bottom-up and top-down tree transformations - a comparison.
\newblock {\em Mathematical Systems Theory}, 9:198--231, 06 1975.
\newblock \href {https://doi.org/10.1007/BF01704020}
  {\path{doi:10.1007/BF01704020}}.

\bibitem{MacroTreeTransd}
Joost Engelfriet and Heiko Vogler.
\newblock Macro tree transducers.
\newblock {\em Journal of Computer and System Sciences}, 31(1):71--146, 1985.
\newblock URL:
  \url{https://www.sciencedirect.com/science/article/pii/0022000085900662},
  \href {https://doi.org/10.1016/0022-0000(85)90066-2}
  {\path{doi:10.1016/0022-0000(85)90066-2}}.

\bibitem{FR1968}
Patrick~C. Fischer and Arnold~L. Rosenberg.
\newblock Multitape one-way nonwriting automata.
\newblock {\em Journal of Computer and System Sciences}, 2(1):88--101, 1968.
\newblock URL:
  \url{https://www.sciencedirect.com/science/article/pii/S0022000068800066},
  \href {https://doi.org/10.1016/S0022-0000(68)80006-6}
  {\path{doi:10.1016/S0022-0000(68)80006-6}}.

\bibitem{GMSZ2008}
Blaise Genest, Anca Muscholl, Olivier Serre, and Marc Zeitoun.
\newblock Tree pattern rewriting systems.
\newblock In Sungdeok~(Steve) Cha, Jin-Young Choi, Moonzoo Kim, Insup Lee, and
  Mahesh Viswanathan, editors, {\em Automated Technology for Verification and
  Analysis}, pages 332--346, Berlin, Heidelberg, 2008. Springer Berlin
  Heidelberg.

\bibitem{Griffiths1968}
T.~V. Griffiths.
\newblock The unsolvability of the equivalence problem for $\lambda$-free
  nondeterministic generalized machines.
\newblock {\em J. ACM}, 15(3):409–413, jul 1968.
\newblock \href {https://doi.org/10.1145/321466.321473}
  {\path{doi:10.1145/321466.321473}}.

\bibitem{Hosoya_2010}
Haruo Hosoya.
\newblock {\em Foundations of XML Processing: The Tree-Automata Approach}.
\newblock Cambridge University Press, 2010.

\bibitem{AK1998}
Armin K{\"u}hnemann.
\newblock Benefits of tree transducers for optimizing functional programs.
\newblock In Vikraman Arvind and Sundar Ramanujam, editors, {\em Foundations of
  Software Technology and Theoretical Computer Science}, pages 146--157,
  Berlin, Heidelberg, 1998. Springer Berlin Heidelberg.

\bibitem{mtchLBound}
Florin Manea and Markus~L. Schmid.
\newblock Matching patterns with variables.
\newblock In Robert Merca{\c{s}} and Daniel Reidenbach, editors, {\em
  Combinatorics on Words}, pages 1--27, Cham, 2019. Springer International
  Publishing.

\bibitem{MSV2000}
Tova Milo, Dan Suciu, and Victor Vianu.
\newblock Typechecking for xml transformers.
\newblock In {\em Proceedings of the Nineteenth ACM SIGMOD-SIGACT-SIGART
  Symposium on Principles of Database Systems}, PODS '00, page 11–22, New
  York, NY, USA, 2000. Association for Computing Machinery.
\newblock \href {https://doi.org/10.1145/335168.335171}
  {\path{doi:10.1145/335168.335171}}.

\bibitem{Nivat68}
Maurice Nivat.
\newblock Transductions des langages de {Chomsky}.
\newblock {\em Annales de l'Institut Fourier}, 18(1):339--455, 1968.
\newblock URL: \url{http://www.numdam.org/articles/10.5802/aif.287/}, \href
  {https://doi.org/10.5802/aif.287} {\path{doi:10.5802/aif.287}}.

\bibitem{DBLP:conf/stoc/OneInThreeSAT}
Thomas~J. Schaefer.
\newblock The complexity of satisfiability problems.
\newblock In Richard~J. Lipton, Walter~A. Burkhard, Walter~J. Savitch, Emily~P.
  Friedman, and Alfred~V. Aho, editors, {\em Proceedings of the 10th Annual
  {ACM} Symposium on Theory of Computing, May 1-3, 1978, San Diego, California,
  {USA}}, pages 216--226. {ACM}, 1978.
\newblock \href {https://doi.org/10.1145/800133.804350}
  {\path{doi:10.1145/800133.804350}}.

\end{thebibliography}


\appendix


\section{Comparison with \texorpdfstring{$k$}{k}-Pebble Transducers}\label{sec:kPebbleTransducerComparison}

We use the notion of \pt{k} as defined in~\cite{MSV2000}. For the sake of completeness, we include the definition here.

\begin{definition}
    A \emph{\pt{k}} is $T := (\alphabet, \alphabet', Q, s, \Delta)$, where
    \begin{itemize}
        \item $\alphabet, \alphabet'$ are ranked input and output alphabets respectively.
        \item $Q$ is a finite set of states.
        \item $s \in Q$ is the initial state.
        \item $\Delta$ is a finite set of transitions, each of which has one of the following forms:
        \begin{itemize}
            \item $(a, \vec{b}, q) \to (q', stay)$,
            \item $(a, \vec{b}, q) \to (q', downLeft)$,
            \item $(a, \vec{b}, q) \to (q', downRight)$,
            \item $(a, \vec{b}, q) \to (q', upLeft)$,
            \item $(a, \vec{b}, q) \to (q', upRight)$,
            \item $(a, \vec{b}, q) \to (q', addPebble)$,
            \item $(a, \vec{b}, q) \to (q', removePebble)$,
            \item $(a, \vec{b}, q) \to (a', output0)$,
            \item $(a, \vec{b}, q) \to (a'(q_1, q_2), output2)$,
        \end{itemize}
        where $a \in \alphabet, a' \in \alphabet', \vec{b} \in \set{0,1}^{i-1}, q, q', q_1, q_2 \in Q$.
    \end{itemize}
\end{definition}

Let $\Gamma$ be a finite alphabet. Let $\rev : \Gamma^* \to \Gamma^*$ be the function defined as follows: For each $w := a_1a_2\dots a_n \in \Gamma^*$ with $a_i \in \Gamma,~\forall i \in [n]$, we let $\rev(w) := {a_n}a_{n-1}\dots a_2a_1$. Let $\Sigma := \Gamma \uplus \set{\#}$ be a finite ranked alphabet with the ranking function $\rank : \Sigma \to \N_0$ given by
\[ \twopartfunc{\rank(a)}{2}{a \in \Gamma}{0}{a = \#} \]
We define the tree-representation-of-string function $\phi : \Gamma^* \to \trees{\alphabet}$ given by defining for $w := a_1a_2\dots a_n \in \Gamma^*$ with $a_i \in \Gamma$, $\phi(w)$ as

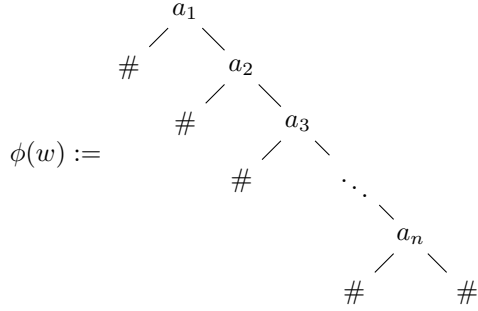
\[
\phi(w) := 
\vcenter{\hbox{\begin{tikzpicture}
    [level distance=3em]
    \node {$a_1$}
    child[grow=south west] { node{$\#$} }
    child[grow=south east] {
        node {$a_2$}
        child[grow=south west] { node{$\#$} }
        child[grow=south east] {
            node{$a_3$}
            child[grow=south west] { node{$\#$} }
            child[grow=south east] {
                node{$\ddots$}
                child[grow=south east] {
                    node{$a_n$}
                    child[grow=south west] { node{$\#$} }
                    child[grow=south east] { node{$\#$} }
                }
            }
        }
    };
\end{tikzpicture}}}
\]

Note that $\phi$ as defined above is an injective function. We define a function $\widetilde{\rev} : \phi[\Gamma^*] \to \phi[\Gamma^*]$ by setting $\widetilde{\rev} := \phi \circ \rev \circ \phi\inv$. By abuse of notation, we shall write $\widetilde{\rev}$ as just \rev\ when there is no need to distinguish between the two.

\begin{theorem}\cite{MSV2000}
    There exists a \pt{2} $T$ such that $\sem{T}(t) := \rev(t)$ for all $t \in \phi[\Gamma^*]$.
\end{theorem}

\begin{proof}
    Consider the \pt{2} $T$ given by $T := (\alphabet, \alphabet, Q, q_0, \Delta)$ where we have $Q := \set{q_0, q_{up}, q_{down}, q_{out}, q_{lleaf}, q_{rleaf}}$. Since there are only 2 pebbles, the bit-vector $\vec{b}$ has atmost one dimension, which is the presence or absence of the first pebble. Hence we shall denote it as $\top$ or $\bot$ respectively whenever applicable. The transition rules in $\Delta$ are as follows:
    \begin{itemize}
        \item $(a, \emptyset, q_0) \to (q_{down}, addPebble),~ \forall a \in \Gamma,$
        \item $(\#, \emptyset, q_0) \to (\#, output0),$
        \item $(a, \top, q_{down}) \to (q_{down}, downRight),~ \forall a \in \Gamma,$
        \item $(a, \bot, q_{down}) \to (q_{down}, downRight),~ \forall a \in \Gamma,$
        \item $(\#, \bot, q_{down}) \to (q_{out}, upLeft),$
        \item $(a, \bot, q_{out}) \to (a(q_{lleaf}, q_{up}), output2),~ \forall a \in \Gamma,$
        \item $(a, \bot, q_{up}) \to (q_{out}, upLeft),~ \forall a \in \Gamma,$
        \item $(a, \bot, q_{lleaf}) \to (q_{out}, downLeft),~ \forall a \in \Gamma,$
        \item $(\#, \bot, q_{out}) \to (\#, output0),$
        \item $(a, \top, q_{up}) \to (a(q_{lleaf}, q_{rleaf}), output2),~ \forall a \in \Gamma,$
        \item $(a, \top, q_{rleaf}) \to (q_{out}, downRight),~ \forall a \in \Gamma,$
    \end{itemize}
    We can observe that the above \pt{2} $T$ does reverse strings, i.e., $\sem{T}(t) = \rev(t)$ for all $t \in \phi[\Gamma^*]$.
\end{proof}

\begin{theorem}\label{thm:bimorphismCannotReverse}
    There does not exist a triple $(L, \psi_1, \psi_2)$ where $L$ is a regular tree language over ranked alphabet $\Lambda$ and $\psi_1, \psi_2 : \trees{\Lambda} \to \trees{\seedAlphabet}$ are tree homomorphisms and a subset $X \subseteq \trees{\Lambda}$ such that $\set{(\psi_1(t), \psi_2(t)) \mid t \in X} \subseteq \rev$.
\end{theorem}

\begin{proof}
We present a proof by contradiction. Let us assume to the contrary that there exists a triple $(L, \psi_1, \psi_2)$ as specified above such that for a subset $X \subseteq \trees{\Lambda}$, we have that $\set{(\psi_1(t), \psi_2(t)) \mid t \in X} \subseteq \rev$. Then we see that when $\psi_1, \psi_2$ as homomorphisms are restricted to the set $\Lambda$, $\psi_1(x), \psi_2(x)$ would necessarily have to be trees of height 1, for all $x \in \Lambda$. In other words, we see that $\psi_1, \psi_2$ are \emph{delabelings}. Hence, from Nivat's Theorem for tree transductions as cited in~\cite{tata}, we know that the bimorphism $(L, \psi_1, \psi_2)$ is equivalent to a bottom-up tree transducer over the alphabet. Since this transducer is by default a finite-state transducer, we see that it cannot realise $\rev$ as a transduced function, since the height of the tree to be reversed is unbounded. This contradicts our assumption about the existence of a bimorphism equivalent to $\rev$, and completes the proof.
\end{proof}

\begin{theorem}
    There does not exist an action $\fullTransform := (S, \psi_\src, \psi_\tgt, \guardFunction)$ such that $\fullTransform = \rev$ on $\phi[\Gamma^*]$.
\end{theorem}

\begin{proof}
    We prove the theorem by contradiction. If possible, let $\fullTransform$ be such that $\fullTransform = \rev$ on $\phi[\Gamma^*]$. We know that by definition, $\fullTransform(t) := \bigcup_{s \in S} \atomicTransform_s(t)$. Let $t_0 \in \phi[{(a+b)}^*]$. Since $\fullTransform(t_0) = \rev(t_0) \neq \emptyset$, let $s_0 \in S$ be such that $\psi_\src(s_0)$ matches $t_0$. Let $\substitution$ be a substitution witnessing the match, i.e., $\substitution(\psi_\src(s_0)) = t_0$.
    
    If there no common variable appearing in $\psi_\src(s_0), \psi_\tgt(s_0)$, then we can substitute multiple contexts from the guard language for a variable appearing in the target pattern $\psi_\tgt(s_0)$ for the same context substituted in place of a variable in the source pattern $\psi_\src(s_0)$. This would contradict the fact that $\rev$ is a bijection.
    
    Let $x$ be a variable in $\psi_\src(s_0)$ that also appears in $\psi_\tgt(s_0)$. Then we see that $\substitution(x) \in \Sigma$. If $\substitution(x) \in \contexts{\Sigma}{}$ has height at least 2, then since $x$ also appears in $\psi_\tgt(s_0)$, $\substitution(\psi_\tgt(s_0))$ can never be $\rev(t_0)$. This implies that $\substitution$ is a letter-to-letter substitution, i.e., $\psi_\src(s_0)$ is a relabeling of $t_0$, and consequently, $\psi_\tgt(s_0)$ is the `reverse' of $\psi_\src(s_0)$. Since this is true for every element in $\phi[{(a+b)}^*]$, we must have that the bimorphism $(S, \psi_\src, \psi_\tgt)$ is such that $\set{(\psi_\src(s), \psi_\tgt(s)) \mid s \in S} \subseteq \rev$, which contradicts Theorem~\ref{thm:bimorphismCannotReverse}.
\end{proof}

\end{document}